\documentclass[11pt]{elsarticle}
\usepackage{amsmath,amsthm,amssymb}
\usepackage{latexsym}
\usepackage{color}
\usepackage{url}
\usepackage{epsfig}
\usepackage{array}
\usepackage{tabularx}
\usepackage{subfigure}
\journal{NeuroImage}

\usepackage{graphicx}
\usepackage{subfigure}
\usepackage{textcomp}
\usepackage{amsmath}
\usepackage{color}
\usepackage{url}
\usepackage{multirow}
\newtheorem{alg}{Algorithm}

\newtheorem{cor}{Corollary}
\newtheorem{theorem}{Theorem}
\newsavebox{\savepar}

\newcommand{\bfx}{\mathbf {x}}
\newcommand{\bfy}{\mathbf {y}}

\newcommand{\bfb}{{\mathbf b}}
\newcommand{\bfa}{{\mathbf a}}

\newcommand{\bfd}{{\mathbf d}}
\newcommand{\bfe}{{\mathbf e}}

\newcommand{\bfp}{{\mathbf p}}
\newcommand{\bfq}{{\mathbf q}}
\newcommand{\bfr}{{\mathbf r}}
\newcommand{\bfs}{{\mathbf s}}

\newcommand{\Cp}{C_{\mathrm{p}}}
\newcommand{\CT}{C_{\mathrm{T}}}
\newcommand{\CF}{C_{\mathrm{F}}}
\newcommand{\CNS}{C_{\mathrm{NS}}}
\newcommand{\CS}{C_{\mathrm{S}}}
\newcommand{\DVT}{\mathrm{DV}_{\mathrm{T}}}

\newtheorem{lemma}{Lemma}

\theoremstyle{remark}

\newcommand{\bitem}{\begin{enumerate}}
\newcommand{\eitem}{\end{enumerate}}
\begin{document}
\begin{frontmatter}
\title{Model Error Correction for Linear Methods of Reversible Radioligand Binding Measurements in PET Studies}
\author[label1]{Hongbin~Guo\corref{cor1}}
\cortext[cor1]{Corresponding author. Tel: 1-480-965-8002,   Fax: 1-480-965-4160.}
\address[label1]{Arizona State University, Department of Mathematics and Statistics, Tempe, AZ  85287-1804. } 
\ead{hguo1@asu.edu}
\author[label1]{Rosemary~A~Renaut} 
%\address[1]{Arizona State University, Department of Mathematics and Statistics, Tempe, AZ  85287-1804. Tel: 1-480-965-3795,   Fax: 1-480-965-4160.}
%\ead{renaut@asu.edu}  
\author[label2]{Kewei~Chen}
\address[label2]{Banner Alzheimer Institute and Banner Good
Samaritan Positron Emission Tomography Center, Phoenix, AZ 85006}
%\ead{Kewei.Chen@bannerhealth.com}
\author[label2]{Eric M Reiman}
%\address{Banner Alzheimer Institute and Banner Good
%Samaritan Positron Emission Tomography Center, Phoenix, AZ 85006}
%\ead{Eric.Reiman@bannerhealth.com}

\begin{abstract}
Graphical analysis  methods are widely used in positron emission tomography  quantification because of their simplicity and model independence. But they may, particularly for reversible kinetics,  lead to  bias in the estimated parameters.  The source  of the bias is commonly attributed to  noise in the data. Assuming a two-tissue compartmental model, we investigate the bias that originates from model error. %associated with model linearization. We find that the bias due to model error may be more significant than that due to noise. 
This bias is an intrinsic property of the simplified linear models used for limited scan durations, and it is exaggerated by  random noise and numerical quadrature error. Conditions are derived under which Logan's graphical method either over- or under-estimates the distribution volume in the noise-free case. The bias caused by model error is quantified analytically. The presented analysis shows that the bias of graphical methods is inversely proportional to the dissociation rate.  Furthermore, visual examination of the linearity of the Logan plot  is not sufficient for guaranteeing that equilibrium  has been reached. A new model which retains the elegant properties of graphical analysis methods is presented, along with a numerical algorithm for its solution. 
We perform simulations with the fibrillar amyloid $\beta$ radioligand [11C] benzothiazole-aniline  using published data from the University of Pittsburgh and Rotterdam groups. %Simulations with the fibrillar amyloid $\beta$ radioligand [11C] benzothiazole-aniline (Pittsburgh Compound-B [PIB]) for data published by Pittsburgh and Rotterdam groups, illustrated  
The results show that the  proposed method significantly reduces the bias due to model error. Moreover, the results for data acquired over a $70$ minutes scan duration are at least as good as those obtained using existing methods for data acquired over  a $90$ minutes scan duration.
\end{abstract}

%\textit{Abbreviated Title:}  Model Error of Graphical Methods
\end{frontmatter}
\section{Introduction}

Graphical analysis (GA) has been routinely used for quantification of positron emission tomography (PET) radioligand measurements. The first GA method  for measuring primarily tracer uptakes for irreversible  kinetics  was introduced by Patlak, \cite{Patlak83,Patlak85}, and extended for measuring tracer distribution (accumulation) in reversible systems by Logan, \cite{Logan90}. These techniques have been utilized both with input data acquired from plasma measurements and using the time activity curve from a reference brain region. They have been used for calculation of tracer uptake rates, absolute distribution volumes (DV) and DV ratios (DVR), or, equivalently, for absolute and relative binding potentials (BP). They are widely used because of their inherent simplicity and general applicability regardless of the  specific compartmental model.

The well-known bias, particularly for reversible kinetics, in parameters estimated by GA is commonly  attributed to noise in the data,   \cite{slifstein2000,ichise2002str,Logan03}, and therefore techniques to reduce the bias have concentrated on limiting the impact of the noise. These include (i) rearrangement of the underlying system of linear equations so as to reduce the impact of noise yielding the so-called \textit{multi-linear} method (MA1), \cite{ichise2002str}, and a second multi-linear approach (MA2), \cite{Blomqvist84}, (ii) preprocessing using the method of generalized linear least squares (GLLS), \cite{Feng96}, yielding a hybrid GLLS-GA method, \cite{logan2001a}, (iii) use of the method of perpendicular least squares, \cite{varga2002mod}, also known as total least squares (TLS), \cite{Golub80}, (iv)  likelihood estimation, \cite{ogden2003est}, (v) Tikhonov regularization \cite{Buchert03}, (vi) principal component analysis, \cite{Joshi_PCA2008}, and (vii)  reformulating the method of Logan so as to reduce the noise in the denominator, \cite{Yun2008bias}. Here, we turn our attention to another important source of the bias: the model  error which is implicit in GA approaches.

The  bias associated with GA approaches  has, we believe, three possible sources. The bias arising due to random noise is most often discussed, but errors may also be attributed to the use of numerical quadrature and  an approximation of the underlying compartmental model. It is demonstrated in Section~\ref{sec:theory} that not only is bias  an intrinsic property of the linear model for limited scan durations, which is exaggerated by noise, but also that it may be  dominated by the effects of the model error. Indeed, numerical simulations, presented in Section~\ref{sec:validatetheory}, demonstrate that large bias can result even in the noise-free case. Conditions for over- or under-estimation of the DV due to model error and the extent of bias of the Logan plot  are quantified analytically. These lead to  the design of a bias correction method, Section~\ref{sec:method}, which still maintains the elegant simplicity of GA approaches. This bias reduction is achieved by the introduction of a simple nonlinear term in the model. While this approach adds some moderate computational expense, simulations reported in Section~\ref{sec:resultsSc0} for the fibrillar amyloid $\beta$ radioligand [11C] benzothiazole-aniline (Pittsburgh Compound-B [PIB]), \cite{Mathis03}, illustrate that it greatly reduces bias. Relevant observations are discussed in Section~\ref{sec:diss} and conclusions  presented in Section~\ref{sec:conc}. %Fundamental mathematical analysis is reserved for the {\em Appendix}.%, while  strict theoretical derivations are reserved for presentation in the {\em Appendix}.

\section{Theory}\label{sec:theory}
\subsection{Existing linear methods}
For the measurement of DV, existing linear quantification methods for reversible radiotracers with a known input function, i.e. the unmetabolized tracer concentration in plasma, are based on the following linear approximation of the true kinetics developed by Logan, \cite{Logan90}: 
\begin{equation}\label{eq:MA0}
\mathrm{ MA0:}\hspace{0.5cm} \int_0^t \CT(\tau){\rm d} \tau \approx \mathrm{DV} \int_0^t\Cp(\tau){\rm d}\tau-b\CT(t).
\end{equation}
Here $\CT(t)$ is the measured \textit{tissue time activity curve} (TTAC), $\Cp(t)$ is 
the \textit{input function}, DV represents the \textit{distribution volume} and quantity $b$ is a  constant. With known $\CT(t)$ and $\Cp(t)$ we can solve for DV and $b$ by the method of linear least squares. This model, which we denote by MA0 to distinguish it from MA1 and MA2 introduced in \cite{ichise2002str},  approximately describes tracer behavior at equilibrium. Dividing through by $\CT(t)$, showing that the DV is the linear slope and $-b$ the intercept,  yields the original   Logan graphical analysis model, denoted here by Logan-GA, 
\begin{equation}\label{eq:Logan}
\mathrm{ Logan-GA:}\hspace{0.5cm}\frac{\int_0^t \CT(\tau){\rm d}\tau}{\CT(t)}\approx DV \frac{\int_0^t\Cp(\tau){\rm d}\tau}{\CT(t)}-b,
\end{equation}
in which the DV and intercept $-b$ are obtained by using linear least squares (LS) for the sampled version of  (\ref{eq:Logan}).  
Although it is well-known that this model often leads to under-estimation of the DV it is still  widely used in  PET studies.  An alternative formulation based on %is obtained by division in (\ref{eq:MA0}) by $b$  instead of $\CT(t)$, yielding model 
(\ref{eq:MA0}) is  the so-called MA1, 
\begin{equation}\label{eq:Multi-lin1}\mathrm{ MA1:}\hspace{0.5cm} \CT(t) \approx \frac{\mathrm{DV}}{b} \int_0^t\Cp(\tau){\rm d}\tau-\frac{1}{b}\int_0^t \CT(\tau){\rm d}\tau,
\end{equation}
for which the DV can again be obtained using LS \cite{ichise2002str}. Recently another formulation, obtained by division in (\ref{eq:MA0}) by $\Cp(t)$  instead of $\CT(t)$, has been developed by Zhou {\it et al}, \cite{Yun2008bias}. But, as noted by Varga {\it et al} in  \cite{varga2002mod} the noise appears in both the independent and dependent variables in (\ref{eq:Logan}) and thus  TLS  may be a  more appropriate model than LS for obtaining the DV. Whereas it has been concluded through numerical experiments for tracer [$^{18}$F]FCWAY and [$^{11}$C]MDL 100,907, \cite{ichise2002str}, that MA1 (\ref{eq:Multi-lin1}) performs better than other linear methods, including Logan-GA (\ref{eq:Logan}), TLS and  MA2 \cite{Blomqvist84, ichise2002str}, none of these techniques explicitly deals with the inherent error due to the assumption of model MA0 (\ref{eq:MA0}).  The focus here is thus examination of the model error specifically for Logan-GA and MA1, from which a new method for reduction of model error is designed. 

\subsection{Model error analysis} \label{subsec:errana}
The general three-tissue compartmental  model for the reversible radioligand binding kinetics of a given brain region or a voxel can be illustrated as follows, \cite{Frost89,Slif01}:

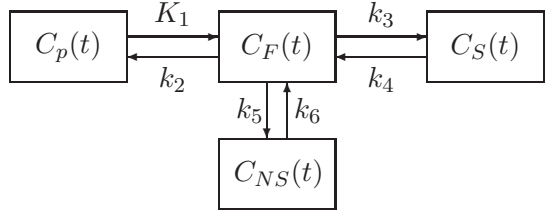
\begin{figure}[h]
\begin{center}
\begin{picture}(9.200000,3.400000)(0.000000,-3.400000)
% picture environment flowchart generated by flow 0.99f
\put(0.0000,-1.2000){\framebox(2.0000,1.2000)[c]{\shortstack[c]{
$C_p(t)$
}}}
\put(2.0000,-1.0000){\makebox(1.6000,0.8000)[c]{\shortstack[c]{
$K_1$\\
~\vspace{2ex}\\
$k_2$
}}}
\put(3.6000,-1.2000){\framebox(2.0000,1.2000)[c]{\shortstack[c]{
$C_F(t)$
}}}
\put(3.6000,-0.7750){\vector(-1,0){1.6000}}
\put(2.0000,-0.4250){\vector(1,0){1.6000}}
\put(5.6000,-1.0000){\makebox(1.6000,0.8000)[c]{\shortstack[c]{
$k_3$\\
~\vspace{2ex}\\
$k_4$
}}}
\put(7.2000,-1.2000){\framebox(2.0000,1.2000)[c]{\shortstack[c]{
$C_S(t)$
}}}
\put(7.2000,-0.7750){\vector(-1,0){1.6000}}
\put(5.6000,-0.4250){\vector(1,0){1.6000}}
\put(3.6000,-2.2000){\makebox(2.0000,1.0000)[c]{\shortstack[c]{
~\hspace{-0.5ex}~\hspace{-1.5ex}  $k_5$ ~\hspace{0.1ex} $k_6$
}}}
\put(3.6000,-3.4000){\framebox(2.0000,1.2000)[c]{\shortstack[c]{
~\hspace{-0.5ex}$C_{NS}(t)$
}}}
\put(4.7750,-2.2000){\vector(0,1){1.0000}}
\put(4.4250,-1.2000){\vector(0,-1){1.0000}}
\end{picture}
\caption{\normalsize Three-tissue compartmental model of reversible radioligand binding dynamics.
 \label{fig:3T}}
\end{center}
\end{figure}

Here $\Cp(t)$ (kBq/ml) is the input function, i.e. the unmetabolized radiotracer concentration in plasma, and $\CF(t)$, $\CNS(t)$ and $\CS(t)$ (kBq/g) are free radioactivity, nonspecific bound and specific bound tracer concentrations, resp., and $K_1$  (ml/min/g) and $k_i$ (1/min), $i=2,\cdots,6$, are rate constants. The DV is related to the rate constants as follows \cite{Gunn01PETmodel},
\begin{equation}\label{eq:DV}
\mathrm{DV}=\frac{K_1}{k_2}(1+\frac{k_3}{k_4}+\frac{k_5}{k_6}). \end{equation}
The numerical implementation for estimating the unknown rate constants of the differential system  illustrated in Figure~\ref{fig:3T}  is difficult  because three exponentials are involved in the solution of this system, \cite{Slif01}. Specifically, without the inclusion of additional prior knowledge, the rate constants may be unidentifiable, \cite{Godfrey83}. Fortunately, for most tracers it can safely be assumed that  $\CNS$ and  $\CF$  reach equilibrium rapidly for specific binding regions. Then it is appropriate to use a two-tissue four-parameter (2T-4k) model by binning $\CNS(t)$ and $\CF(t)$ to one compartment $C_{F+NS}(t)=\CF(t)+\CNS(t)$. This  is equivalent to taking $k_5=k_6=0$, and hence $\CNS(t)=0$.  On the other hand, for regions without specific binding activity, we know $\CS(t)=0$ which is equivalent to taking $k_3=k_4=0$, and it is again appropriate for most radioligands to bin  $\CNS(t)$ and $\CF(t)$. The one-tissue compartmental model is then appropriate for regions without specific binding activity. For some tracers, however, for example  the modeling of PIB in the cerebellar reference region, the best data  fitting is obtained by using the 2T-4k model without binning $\CNS(t)$ and $\CF(t)$, \cite{Price05}.   Assuming the latter, the  DV is given by $K_1/k_2(1+k_3/k_4)$, and $K_1/k_2(1+k_5/k_6)$, for regions with and without specific binding activity, resp. Ignoring the notational differences between the two models,  for regions with and without specific binding activity, they are both described by the same abstract mathematical 2T-4k model   equations. Here, without loss of generality,  we present the 2T-4k  model equations for specific binding regions, 
\begin{eqnarray}
 \frac{{\rm d}C_{F+NS}(t)}{{\rm d}t}&=&K_1\Cp(t)-(k_2+k_3)C_{F+NS}(t)+k_4\CS(t) \label{eq:2T4kC1} \\
 \frac{{\rm d}\CS(t)}{{\rm d}t}&=&k_3C_{F+NS}(t)-k_4\CS(t). \label{eq:2T4kCS}
\end{eqnarray}
To obtain the equations appropriate for regions without specific binding activity, $\CS(t)$ is replaced by  $\CNS(t)$ and $k_3$ and $k_4$ are interpreted as the association and dissociation parameters of  regions without specific binding activity. To simplify the explanation  $\CS(t)$, $k_3$ and $k_4$ are used throughout for both regions with and without specific binding activity, with the assumption that $\CS(t)$, $k_3$ and $k_4$ should automatically be replaced by $\CNS(t)$, $k_5$ and $k_6$ respectively, when relevant.  

The solution of the linear differential system (\ref{eq:2T4kC1})-(\ref{eq:2T4kCS})
is given by 
\begin{eqnarray}
 C_{F+NS}(t)&=& ( a_1e^{-\alpha_1t}+b_1e^{-\alpha_2t}   ) \otimes \Cp(t)\label{sol:C_1} \\
 \CS(t)&=&a_2( e^{-\alpha_1t}-e^{-\alpha_2t} ) \otimes \Cp(t)   \label{sol:CS}
\end{eqnarray}
where $\otimes$ represents the convolution operation, 
\begin{eqnarray}
\alpha_{1, 2} &=& ( k_2+k_3+k_4 \mp \sqrt{(k_2+k_3+k_4)^2 - 4k_2 k_4}\ ) /2, \quad \mathrm{and}\nonumber \\
a_1&=& \frac{K_1(k_4-\alpha_1)}{\alpha_2-\alpha_1},\,  b_1=\frac{K_1(\alpha_2-k_4)}{\alpha_2-\alpha_1},\, \mathrm{and} \, a_2= \frac{K_1k_3}{\alpha_2-\alpha_1}. \label{symb:a1}
\end{eqnarray}
The overall concentration of radioactivity  is
\begin{equation}\label{sol:CT}
\CT(t)=C_{F+NS}(t)+\CS(t)= (( a_1+a_2)e^{-\alpha_1t}+(b_1-a_2)e^{-\alpha_2t}   ) \otimes \Cp(t). 
\end{equation}
Integrating (\ref{eq:2T4kC1})-(\ref{eq:2T4kCS}) and rearranging yields
\begin{eqnarray}
 \int_0^t\CT(\tau){\rm d}\tau&=&\mathrm{DV}\int_0^t\Cp(\tau){\rm d}\tau-\frac{k_3+k_4}{k_2k_4}C_{F+NS}(t)-\frac{k_2+k_3+k_4}{k_2k_4}\CS(t),\label{eq:sfunc}\\
&=&\mathrm{DV}\int_0^t\Cp(\tau){\rm d}\tau-\frac{k_3+k_4}{k_2k_4}\CT(t)-\frac{1}{k_4} \CS(t). \label{eq:root}
\end{eqnarray}
This is model (\ref{eq:MA0}) when $\CS(t)$ is linearly proportional to $\CT(t)$ for a time window within the scan duration of $T$ minutes. The accuracy of  linear methods based on (\ref{eq:MA0}) is thus dependent on the validity of the assumption that  $\CS(t)$ and $C_{F+NS}(t)$ are approximately linearly proportional to $\CT(t)$ over a time window within $[0,T]$. %For the estimation of the tracer uptake rate in irreversible systems, Patlak assumed that both $C_{F+NS}(t)$ and $\CS(t)$ are  almost proportional to $\Cp(t)$ after some time $t'$, \cite{Patlak83}. On the other hand, 
Logan observed that  $C_{F+NS}(t)$ and $\CS(t)$ are roughly proportional to $\CT(t)$, after some time point  $t^*$, \cite{Logan90}.  If the assumption of linear proportionality breaks down for the given window,  $[t^*, T]$,  bias in the estimated uptake rate or DV will be introduced, as shown later in Section~\ref{sec:resultsSc0},  due to the intrinsic model error of a GA method. % will introduce bias to the estimated DV, amounts of which could be significant as we demonstrated. %Therefore, the bias is intrinsic to the model and can only be removed when the system reaches satisfactory equilibrium. 
Indeed, in  Section \ref{subsec:equilibrium} we show that, for the PIB radioligand on some regions with small $k_4$, there is no window within a $90$ minutes scan duration  where $\CS(t)$ and $\CT(t)$ are linearly  proportional. This is despite the apparent good linearity, visually, of the  Logan plot  of $\int_0^t \CT(\tau){\rm d}\tau/\CT(t)$ against $\int_0^t\Cp(\tau){\rm d}\tau/\CT(t)$. Waiting for  equilibrium, which may take several hours, is impractical in terms of   patient comfort, cost and measurement of radioactivities. 

The limitation of the constant approximation can be analysed theoretically. Because $\alpha_2>>\alpha_1>0$ and $\Cp(t)$ is very small for large time the convolution  $e^{-\alpha_2t}\otimes \Cp(t)=\int_0^te^{-\alpha_2(t-\tau)}\Cp(\tau){\rm d}\tau$ % which is roughly  the sum of products of tiny part of $e^{-\alpha_2t}$ with early part of $\Cp(t)$ and small part of $\Cp(t)$ with early part of $e^{-\alpha_2t}$, 
is relatively small. We can safely assume that the ratio of $e^{-\alpha_2t}\otimes \Cp(t)$ to $e^{-\alpha_1t}\otimes \Cp(t)$ is roughly $0$ for $t>t^*$.  Then $\CS(t)$, see equation (\ref{sol:CS}), is approximately proportional to $e^{-\alpha_1t} \otimes \Cp(t)$ for $t>t^*$. In our tests with PIB, the neglected component  $a_2e^{-\alpha_2t} \otimes \Cp(t)$ is less than $ 8\%\CS(t)$ for $t\ge 35$ min.. On the other hand, this is not the case for $C_{F+NS}(t)$, see equation (\ref{sol:C_1}), because $a_1$ and $b_1$ need not be of the same scale. For example, if $k_4<<k_2+k_3$ we know $b_1/a_1\approx (k_2+k_3)/(2k_4)$ from (\ref{symb:a1}), thus $b_1>>a_1>0$. Specifically,  $b_1e^{-\alpha_2t} \otimes \Cp(t)$ may not be small in relation to  $a_1e^{-\alpha_1t} \otimes \Cp(t)$. Thus, it is not appropriate, as is assumed for the Logan-GA (\ref{eq:Logan}) and other linear methods derived from MA0, to approximate 
\begin{equation}\label{b_range}
\bar{\bfs}(t)=\frac{k_3+k_4}{k_2k_4}\cdot \frac{C_{F+NS}(t)}{\CT(t)}+\frac{k_2+k_3+k_4}{k_2k_4}\cdot \frac{\CS(t)}{\CT(t)},
\end{equation}
as constant for $t\in [t^*, T]$. One may argue that if $(a_1+a_2)/(b_1-a_2)$ is close to $1$  the term  $e^{-\alpha_2t} \otimes \Cp(t)$ in $\CT(t)$ could be ignored. Then the ratio of $\CT(t)$ to $\CS(t)$ would be close to constant after $t^*$, and the resulting estimates of the DV using Logan-GA (\ref{eq:Logan}) and MA1 (\ref{eq:Multi-lin1}) would be reasonable.  While it is easy to verify that $(a_1+a_2)/(b_1-a_2)$ is positive and bounded above by one, this fraction need not be close to its upper bound. Indeed, for realistic test data, see Table~\ref{tab:ki},  $0.05\le (a_1+a_2)/(b_1-a_2)\le 0.65$.  The simulations presented in Tables~\ref{tab:overest} and \ref{tab:Sc0} validate that a small value of this fraction may cause a problem in the estimation of the DV using the linear Logan-GA and MA1 methods.

It is immediate using $\CT(t)=C_{F+NS}(t)+\CS(t)$, and positivity of both $C_{F+NS}(t)/\CT(t)$ and $\CS(t)/\CT(t)$, that $\bar{\bfs}(t)$ is bounded above and below, 
\begin{equation} \label{ineq:bounds_s}
\frac{k_3+k_4}{k_2k_4}<  \bar{\bfs}(t)< \frac{k_2+k_3+k_4}{k_2k_4}=\frac{k_3+k_4}{k_2k_4}+\frac{1}{k_4},
\end{equation}
and $1/k_4$ determines the variation in $\bar{\bfs}(t)$. If $k_4$ is  small the bound is not tight and the  DV estimated by Logan-GA, or a linear method derived from MA0, may not be accurate, see for example the regions of interest (ROIs) \textbf{1}, \textbf{3} and \textbf{6} in the test examples reported in Table~\ref{tab:ki}. We reiterate that, by the discussion above, the variation for ROI~\textbf{6}, within which no specific binding activity exist, is   determined by $1/k_6$. 
% Noting now that (\ref{eq:2T4kCS}) implies 
% $$\CS=k_3 e^{-k_4 t} \otimes C_{F+NS}(t),$$
% it is apparent that the ratio $\CS(t)$ to $C_{F+NS}(t)$ is not constant in finite time when $k_4$ is very small. On the other hand, when $k_4$ is large, $e^{-k_4 t}$  behaves like an impulse function which guarantees that $\CS(t)$ is proportional to $C_{F+NS}(t)$ after a very short time interval.  
This relationship between the size of $k_4$ and the bias in the Logan-GA estimate of the DV  is illustrated in Figure~\ref{fig:Bias_k4} of Section~\ref{subsec:equilibrium} for the test data of Table~\ref{tab:ki}. %We illustrate the relation between the bias of the Logan-GA model and the  size of  $k_4$ in Section \ref{ } for our testing examples.

\subsection{Model error of Logan equation}

The complete mathematical result for the  model error of Logan-GA and MA0 is presented  in the  {\em Appendix}. Similar results, omitted here to save space,  can be obtained for MA1. The main conclusion is that both Logan-GA and MA0 can lead to an over-estimation of the DV. This contrasts the standard view of these methods. We summarize in the following theorem, for which %the entire proof is given in the {\em Appendix}. 
the main idea is to show that replacing  (\ref{b_range}) which occurs on the right hand side of (\ref{eq:sfunc}) by a constant intercept $b$ introduces an error in the least squares solution for the DV which can be specifically quantified.   
\begin{cor}\label{thm:Logan_bias}
Suppose Logan-GA, or respectively MA0,  are used for noise-free data acquired for $n$ frames with frame time $t_i, i=1,\cdots,n$ and $t^*=t_l$. Then, with $\bar{\bfs}(t)$ as defined in (\ref{b_range}),  for each method the same conclusions are reached:
\begin{itemize}
 \item The DV is over-estimated (under-estimated)  if $\bar{\bfs}(t), t \in [t_l, t_n],$  is a  non-constant decreasing (increasing) function, and 
%\item Logan-GA (MA0) under-estimate DV  if $\bfs(t)/\CT(t), t \in [t_p, t_n],$ is a non-constant increasing function,
\item  the DV is exact if $\bar{\bfs}(t), t \in [t_l, t_n],$  is a constant function;
\end{itemize}
Let $\mathrm{DV}_\mathrm{T}$ be the true value of the DV, and define the variation of a function over $[t_l, t_n]$ by  
\begin{equation}\label{eq:var}
V(\bfx(t))=|\displaystyle\max_{t \in [t_l, t_n]}\bfx(t)-\displaystyle\min_{t \in [t_l, t_n]} \bfx(t)|.
\end{equation} Then the bias in $\mathrm{DV}_\mathrm{L}$ calculated by Logan-GA is bounded by 
{\small
\begin{equation}\label{LoganDV_bound}
 |\mathrm{DV}_\mathrm{L}-\mathrm{DV}_\mathrm{T}|\le \frac{(n-l+1)\displaystyle\sum_{i=l}^n\bar{\bfp}_i}{
\displaystyle\sum_{i\neq j, l\le i, j\le n} (\bar{\bfp}_i-\bar{\bfp}_j)^2
} V(\bar{\bfs}(t)),
\end{equation}
}
where $\bar{\bfp}_i=\int_0^{t_i}\Cp(\tau){\rm d}\tau/\CT(t_i)$.% and $\bar{\bfs}(t)$ is defined in (\ref{}).
\end{cor}

This theorem is an immediate result of Lemma~\ref{lemma:MA1_ana} and  Corollary~\ref{cor:Logan-GA} in the  {\em Appendix} for the vectors obtained from the sampling of the functions 
$$\bfs(t)=\frac{k_3+k_4}{k_2k_4}C_{F+NS}(t)+\frac{k_2+k_3+k_4}{k_2k_4}\CS(t),\,\, \mathrm{and}$$ 
$$\bfr(t)=\int_0^t\CT(\tau){\rm d}\tau,\quad  \bfp(t)=\int_0^t\Cp(\tau){\rm d}\tau, \quad \bfq(t)=\CT(t),$$
at discrete time points $t=t_l,\cdots, t_n.$  The relevant vectors are defined by 
$\bar{\bfr}=\bfr/\bfq$, $\bar{\bfp}=\bfp/\bfq$,  $\bar{\bfs}=\bfs/\bfq$, where the division corresponds to componentwise division. It is easy to check that all these vectors are positive vectors, $\bfp$, $\bar{\bfp}$, $\bfr$ and $\bar{\bfr}$ are non-constant increasing vectors and $\bfq$ is decreasing. Thus all conditions for Lemma~\ref{lemma:MA1_ana} and  Corollary~\ref{cor:Logan-GA} are satisfied. Note that in the denominator of (\ref{LoganDV_bound}) the simplification $(n-l+1)\sum_{i=l}^n(\bar{\bfp}_i)^2-(\sum_{i=l}^n\bar{\bfp}_i)^2=\sum_{i\neq j, l\le i, j\le n} (\bar{\bfp}_i-\bar{\bfp}_j)^2$ is used. In the latter discussion we may use the variation (increasing or decreasing) of $\CS(t)/\CT(t)$ instead of that of $\bar{\bfs}(t)$ because 
$$\bar{\bfs}(t)=\frac{k_3+k_4}{k_2k_4}+\frac{1}{k_4}\CS(t)/\CT(t).$$

It is not surprising that the properties of Logan-GA and MA0 are similar. Indeed, MA0 is none other than weighted Logan-GA with weights $\CT(t_i)$, which changes the noise structure in the variables. In contrast to the conventional under-estimation observations, it is suprising that the DV may be over-estimated.  However, the over-estimation is indeed observed in the tests  presented in Section \ref{sec:overest} and \ref{sec:resultsSc0}. Inequality (\ref{LoganDV_bound}) indicates that Logan-type linear methods will work well for data for which  $V(\bar{\bfs})$ is flat.  Unfortunately,  $V(\bar{\bfs})$ may become flat only for a late time interval. Thus our interest, in Section~\ref{sec:method}, is to better estimate the DV using a reasonable (practical) time window, which may include the window over which $\CS(t)/\CT(t)$ is still increasing. Our initial focus is on the modification of Logan-type methods. Then, in Section~\ref{sec:validatetheory} we present numerical simulations using noise-free data which illustrate the difficulties with Logan-GA and MA1, and support the results of Theorem~\ref{thm:Logan_bias}.

\section{Methods} \label{sec:method}
In the previous discussion we have seen the theoretical limitations of the Logan-GA and MA1 methods. Here we present a new model and associated algorithm which assists with reducing the bias in the estimation of the DV. 

%\subsection{A refined model and algorithm}\label{model:new}
Observe that, $\alpha_2>> \alpha_1$, implies that $C_S=a_2e^{-\alpha_1t} \otimes \Cp(t)+\epsilon(t)$, where $\epsilon(t)$ can be ignored for $t>t^*$. Therefore, for $t>t^*$  (\ref{eq:root})  can be approximated by a new model as follows
\begin{equation}\label{eq:root1}
 \int_0^t\CT(\tau){\rm d}\tau \approx DV\int_0^t\Cp(\tau){\rm d}\tau-A\CT(t)-Be^{-\alpha_1t} \otimes \Cp(t),
\end{equation}
where $A={(k_3+k_4)}/{k_2k_4}$ and $B={a_2}/{k_4}$. This suggests new algorithms should be developed for estimation of parameters DV, $A$, $B$ and $\alpha_1$. Here, a new approach, based on  the basis function method (BFM)  in \cite{Gunn97BFM}, in which $\alpha_1$ is discretized, is given by the following Algorithm.

\begin{alg}\label{alg:Biascor_Guo}Given $\Cp(t_i)$ and $\CT(t_i)$ for $i=1,\cdots,n$ and $t^*=t_l$, the DV is estimated by performing the following steps.
\begin{enumerate}
\item Calculate  $\mathrm{DV}$ and intercept $-b$, using Logan-GA. 
\item Set $\alpha_1^{\mathrm{min}}=0.001$ and $\alpha_1^{\mathrm{max}} =\min(1,2/b)$ if $b>0$ otherwise $\alpha_1^{\mathrm{max}}=1$. 
\item Form discretization $\alpha_1^{(j)}$, $j=1:100$ for  $\alpha_1$, with equal spacing logarithmically between $\alpha_1^{\mathrm{min}}$ and $\alpha_1^{\mathrm{max}}$. 
\item For each $j$ solve the linear LS problem, i.e. cast it as a multiple linear regression problem with $\int_0^t C_T(\tau) \, d\tau$ as the dependent variable.
\begin{equation}
 DV\int_0^{t}\Cp(\tau){\rm d}\tau-A\CT(t) -B \int_0^{t} e^{-\alpha^{(j)}_1 \tau} \Cp(t-\tau){\rm d}\tau\approx \int_0^{t}\CT(\tau){\rm d}\tau \label{eq:LSnewmodel}
\end{equation}
 with data at $t_i$, $i=l,\cdots,n$, to give values  $\mathrm{DV}^{(j)}$, $A^{(j)}$ and $B^{(j)}$.: 
\item Determine $\alpha_1^{(j^*)}$ for which residual is minimum over all $j$. Set  $\mathrm{DV}$,  $A$ and $B$  to be  $\mathrm{DV}^{(j^*)}$, $A^{(j^*)}$ and $B^{(j^*)}$, resp. 
\end{enumerate}
\end{alg}

{\bf Remarks:}\\
\begin{enumerate}
\item \label{remark:alpha1} The interval for $\alpha_1$ is determined as follows: First the lower bound $0.001$ for $\alpha_1$ is suitable for most tracers, but could be reduced appropriately.  This lower bound is not the same as that on  $\theta$ used in BFM, in which $\theta$ is required to be greater than the decay constant of the isotope, \cite{Gunn97BFM}.  Second by point (\ref{item:est_b}) of Corollary~\ref{cor:Logan-GA} in  the {\em Appendix A}, $b$ should be positive and near the average value of $\bar{\bfs}(t)$, where, by (\ref{ineq:bounds_s}), $\frac{k_3+k_4}{k_2k_4}<\bar{\bfs}(t)<\frac{k_2+k_3+k_4}{k_2k_4}$. On the other hand,   $\frac{k_2+k_3+k_4}{k_2k_4}\approx \frac{1}{\alpha_1}$ if $4k_2k_4$ is small relative to $(k_2+k_3+k_4)^2$. Thus, $\alpha_1$ is linked with $b$ through $\bar{\bfs}(t)$. This is used to give the estimate of the upper bound on $\alpha_1$. Practically, it is possible that the Logan-GA may yield an intercept $b<0$, then  we set $\alpha_1^{\mathrm{max}}=1$. 

\item  Numerically, because $\int_0^{t}\Cp(\tau){\rm d}\tau$ is much larger than both $\CT(t)$ and $\CS(t)$ for $t>t^*$, the estimate of DV is much more robust to noise in the formulation, including both model and random noise effects, than are the estimates of $A$ and $B$.  Therefore, while $A$ and $B$ may not be good estimates of $(k_3+k_4)/(k_2k_4)$ and $a_2/k_4$, resp. for noisy data, the estimate of DV will still be acceptable. Consequently, it is possible that Logan-GA and MA0 will produce reasonable estimates for DV, even when the model error is non negligible.   %$A$ may be even a negative number. However, we can still have an ```` estimation for DV for noisy cases. This is also the reason why Logan-GA and MA1 can produce DVs with reasonable error for most applications even the model error may not be small.
%\item It is possible that the solution with minimum residual at the final stage of the algorithm occurs at an outlier (unrealistic) solution. This is avoided by seeking the solution with minimal residual within a bracket around $\mathrm{DV}_{\mathrm{L}}$. Here we use the bracket  $0.8\mathrm{DV}_{\mathrm{L}}\le \mathrm{DV} \le 1.35\mathrm{DV}_{\mathrm{L}}$, which could  be empty. If the set is empty the solution is chosen as that for which the resulting DV is closest to the upper limit on the bracket,   $1.35\mathrm{DV}_{\mathrm{L}}$. This situation, which occurs in the simulations only about $.1\%$ of the time for high noise,  arises when $\mathrm{DV}_{\mathrm{L}}$ extremely under-estimates the true DV. The choice of the scale factor  $1.35$ in the upper bound is a conservative choice. Practically, for the cases in which under-estimation is observed,  $\mathrm{DV}_{\mathrm{L}}$ is  about one half of the true analytically-calculated value.% about half of the true value 
\item The algorithm can be accelerated by employing a coarse-to-fine multigrid strategy. The coarser level grid provides bounds for the fine level grid. The grid resolution can be gradually refined until the required accuracy is satisfied.
%\item The algorithm can be implemented to calculate the DV for any number of TTACs simultaneously, or for generating a parametric image of the DV.  This can significantly improve the algorithmic efficiency. 
\end{enumerate}

\section{Experimental Results}\label{sec:validatetheory}
We present a series of simulations which first validate the theoretical analysis of Section~\ref{sec:theory} for noise-free data, and then numerical experiments which contrast the performance of  Algorithm~\ref{alg:Biascor_Guo} with  Logan-GA,  MA1 and nonlinear kinetic analysis (KA) algorithms for noisy data.
\subsection{Simulated Noise-Free Data}\label{sec:bias}
We assume the radioligand binding system is well modeled by the 2T-4k compartmental model and focus the analysis on the bias in the estimated DV which can be attributed to the simplification of the 2T-4k model. For the simulation we use representative kinetic parameters for brain studies with the PIB tracer. These kinetic parameters, detailed in Table~\ref{tab:ki}, are adopted from published clinical data, \cite{Price05, Yaqub08PIB}. The simulated regions include the posterior cingulate (PCG), cerebellum (Cere) and a combination of cortical regions (Cort). The kinetic parameters of each ROI are also associated with the subject medical condition, namely   normal controls (NC) and  Alzheimer's Disease (AD) diagnosed subjects. The kinetic parameters for the first seven ROIs are from 
\cite{Price05} while the last four are from \cite{Yaqub08PIB}. Rate constants for ROIs~\textbf{5} to ~\textbf{11}  are directly adopted from the published literature, while those for  ROIs~\textbf{1} to \textbf{4} are rebuilt from information provided in \cite{Price05}. The values for ROIs~\textbf{1} to \textbf{4} and \textbf{8} to \textbf{11} represent average values for each group, while those for ROIs~\textbf{5} and \textbf{6} are derived from one AD subject and those for ROI~\textbf{7} from another AD subject. 

%\newpage
\begin{table}[htbp]
\begin{center}
\caption{Rate constants for eleven ROIs, including PCG,  Cere, and Cort,  for AD and NC  adopted from \cite{Price05, Yaqub08PIB}. For ROIs~\textbf{6}, \textbf{7}, \textbf{10} and \textbf{11} no specific binding activity is assumed,  i.e. $k_3=k_4=0$, $DV=K_1/k_2(1+k_5/k_6)$; while for ROIs~\textbf{1} to \textbf{5}, \textbf{8} and \textbf{9}  we assume that the free and nonspecific compartments rapidly reach equilibrium, i.e. $k_5=k_6=0$,  $DV=K_1/k_2(1+k_3/k_4)$. Coefficients $a_1, b_1$ and $a_2$ are defined in (\ref{symb:a1}). The values for ROIs~\textbf{1} to \textbf{4} and \textbf{8} to \textbf{11} represent average values for each group, while those for ROIs~\textbf{5} and \textbf{6} are derived from one AD subject and those for ROI~\textbf{7} from another AD subject.  \label{tab:ki}}
\begin{tabular}{|c|c|c|c|c|c|c|c|c|c|}
\hline
\textbf{ROI/Group} &Area &\textbf{$K_1$}&\textbf{$k_2$}&\textbf{$k_3$ }&\textbf{$k_4$}&\textbf{$k_5$ }&\textbf{$k_6$}&\textbf{DV}& $\frac{a_1+a_2}{b_1-a_2}$\\\hline
\textbf{1}/\textbf{NC}&\textbf{Cort}&0.250&0.152&0.015&0.0106&0&0&3.9722&0.11 \\\hline
\textbf{2}/\textbf{AD}&\textbf{Cort}&0.220&0.113&0.056&0.023&0&0&6.6872&0.65\\\hline
\textbf{3}/\textbf{NC}&\textbf{PCG }&0.250&0.150&0.015&0.0106&0&0&4.0252&0.11\\\hline
\textbf{4}/\textbf{AD}&\textbf{PCG }&0.220&0.100&0.050&0.017&0&0&8.6706&0.63\\\hline
\textbf{5}/\textbf{AD}&\textbf{PCG }&0.262&0.121&0.044&0.015&0&0&8.5168&0.44\\\hline
\textbf{6}/\textbf{AD}&\textbf{Cere}&0.273&0.144&0&0&0.007&0.005&4.5500&0.05\\\hline
\textbf{7}/\textbf{AD}&\textbf{Cere}&0.333&0.172&0&0&0.029&0.042&3.2728&0.26\\\hline
\textbf{8}/\textbf{NC}&\textbf{Cort}& 0.250& 0.140&0.020&0.018&0&0& 3.7480& 0.18\\\hline
\textbf{9}/\textbf{AD}&\textbf{Cort}&0.220 & 0.110&0.050&0.025&0&0& 5.9841 &0.63\\\hline
\textbf{10}/\textbf{NC}&\textbf{Cere}&0.270&0.140&0&0& 0.020&0.026& 3.4353&0.20\\\hline
\textbf{11}/\textbf{AD}&\textbf{Cere}&0.260&0.130&0&0& 0.020&0.025& 3.5810&0.22\\\hline
\end{tabular}
\end{center}
\end{table}
%\newpage

The noise-free decay-corrected  input function is adapted from the plasma measurements for a NC subject as   presented in  Figure~3(A) of \cite{Price05}. Using  the data from that figure we convert to kBq/ml under the assumption of a $100$kg body mass, and obtain the functional representation  for $\Cp(t)=u(t)$, (kBq/ml), which is illustrated in Figure~\ref{fig:inpf}:
\begin{equation}\label{inpf.eq}
u(t)=\left\{\begin{array}{ll}
0 & t\in[0,0.3]\\
407.4933(t-0.3)& t \in [0.3, 0.6]\\  %adjusted 33.8934 to 33.903 so that two line meet at the rig
-436.6t+  384.208 & t \in [0.6, 0.76]\\
%26.196((t+0.24)^{-12.4706}+ (t+0.24)^{-0.9980})  & t \ge 0.76.
46.6747(t+0.24)^{-2.2560}+ 5.7173(t+0.24)^{-0.5644}  & t \ge 0.76.
\end{array}\right.
\end{equation}

\begin{figure}[htbp]
\centerline{\includegraphics[scale=.45]{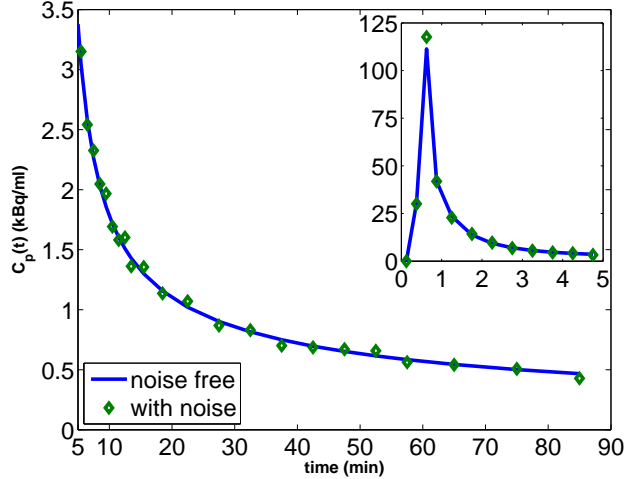}}
\caption{The true input function as given by (\ref{inpf.eq}), and the simulated measurements with noise. The simulated measurements are generated by (\ref{eq:noisyinput}) with C$V_\mathrm{S}=0.05$,  $e=50\%$, $\mu=0.5$ml and  $\Delta w_i=100$ seconds. The function over the initial 5 minutes is illustrated in the inset.
\label{fig:inpf}}
\end{figure}

Using this input function and the eleven data sets given in Table~\ref{tab:ki} eleven noise-free TTACS, $\CT(t)$ (kBq/ml),  are generated using the 2T-4k model. The scanning protocol, consistent with that adopted in \cite{Price05}, has frame durations, $\Delta t_i$, measured in minutes, $4\times 0.25$, $8\times 0.5$, $9 \times 1$, $2\times 3$, $8\times 5$ and $3\times 10$.  The last eight frames, which fall in the window from  $35$ to $90$ minutes, are chosen for the time window over which we assume that equilibrium is achieved. A scan duration of $90$ minutes is common for most PIB-PET dynamic studies, \cite{Muetal:05}.

%\subsection{Bias in the Noise-Free Case}\label{subsec:bias_source}
%this subsection  is split and moved to next subsection and discussion section.

%\subsection{Equilibrium Behavior}
%this subsection is moved to discussion

\subsection{Examples of over-estimation for Logan-GA and MA1}\label{sec:overest}
Theorem \ref{thm:Logan_bias} predicts that the DV will be over-estimated when $\bar{\bfs}$  decreases.  
This is validated for  data for the  simulated ROIs. The estimates of the DV, for scan durations $T=90$ minutes with $t^*=35$ minutes, and $T=240$ minutes with $t^*=100$ minutes, are reported in Table~\ref{tab:overest}. The extended time window is generated by adding $15$ frames each of $10$ minutes length. Indeed, the over-estimation predicted in Theorem~\ref{thm:Logan_bias} is confirmed for ROI~\textbf{7}, for which the decrease of $\CS(t)/\CT(t)$ and, hence $\bar{\bfs}$ after $35$ minutes, is clearly illustrated in  Figure \ref{fig:CS2CT}. Moreover,  $\CS(t)/\CT(t)$ is decreasing after $100$ minutes for all ROIs except ROI~\textbf{6}, see Figure~\ref{fig:CS2CT720}, and in all but this case the values of DV are over-estimated. We note that $\bar{\bfs}$ is nearly flat on the selected windows, $[t^*, T]$ for the cases in which the over-estimation of DV is small. These results further validate the conclusions of Theorem~\ref{thm:Logan_bias}. Additionally, the use of the long scan duration of $240$ minutes leads to estimates with less overall bias because the variation in $\CS(t)/\CT(t)$ is smaller over $[100 min., 240 min.]$ than over the earlier window. Equivalently, as given by (\ref{LoganDV_bound}), a  small variation in $\bar{\bfs}$ guarantees a small error in the estimated DV. Clearly, linear methods based on the MA0 model work well during the equilibrium phase. Unfortunately, this equilibrium may be reached too late for practical application, see for example ROI~\textbf{6} in Figure~\ref{fig:CS2CT720}, for which approximate equilibrium is not reached until $3$ hours. The results with $90$ minutes scan duration show that better estimates are obtained for larger  $(a_1+a_2)/(b_1-a_2)$, which consistently supports the analysis in Section~\ref{subsec:errana}. % after equation (\ref{b_range}).

In these simulations the accurate data and integrals are used so as to assure that the results are not impacted by use of a low accuracy numerical quadrature but instead are focused on the effects of the model error of Logan-GA and MA1.  It is  interesting to note, however, that the error introduced by the numerical quadrature always lowers the estimate of  the DV, see Section~\ref{subsec:quad_effect}. Moreover, the noise from other sources may  have a similar impact. This is a topic for future research.

%\newpage
\begin{table}[htbp]
\begin{small}
\begin{center}
\caption{The DV calculated using  Logan-GA and MA1 with noise-free data and accurate integrals. DV is calculated  for scan durations  $T=90$ minutes with $t^*=35$ minutes, and $T=240$ minutes with $t^*=100$ minutes. The percentage bias is listed in parentheses. \label{tab:overest}}
\begin{tabular}{|c|c|c|c|c|c|}
\hline
ROI & True & \multicolumn{2}{|c|}{$35$-$90$ min} &\multicolumn{2}{|c|}{$100$-$240$ min} \\ \cline{3-6}
ID  & DV & Logan-GA & MA1 & Logan-GA &MA1 \\\hline

\textbf{1} & $3.9722$  & $3.549$($-10.65$\%)& $3.542$($-10.84$\%)& $3.981$($0.22$\%)& $3.977$($0.12$\%)\\ \hline 
\textbf{2} & $6.6872$  & $6.585$($-1.53$\%)& $6.577$($-1.65$\%)& $6.709$($0.33$\%)& $6.709$($0.33$\%)\\ \hline 
\textbf{3} & $4.0252$  & $3.593$($-10.73$\%)& $3.586$($-10.92$\%)& $4.034$($0.22$\%)& $4.030$($0.11$\%)\\ \hline 
\textbf{4} & $8.6706$  & $8.342$($-3.79$\%)& $8.331$($-3.92$\%)& $8.687$($0.19$\%)& $8.685$($0.16$\%)\\ \hline 
\textbf{5} & $8.5168$  & $8.129$($-4.55$\%)& $8.117$($-4.69$\%)& $8.536$($0.23$\%)& $8.533$($0.19$\%)\\ \hline 
\textbf{6} & $4.5500$  & $3.204$($-29.58$\%)& $3.208$($-29.50$\%)& $4.281$($-5.91$\%)& $4.273$($-6.10$\%)\\ \hline 
\textbf{7} & $3.2728$  & $3.300$($0.82$\%)& $3.298$($0.76$\%)& $3.286$($0.41$\%)& $3.288$($0.45$\%)\\ \hline 
\textbf{8} & $3.7480$  & $3.635$($-3.01$\%)& $3.625$($-3.28$\%)& $3.780$($0.84$\%)& $3.779$($0.84$\%)\\ \hline 
\textbf{9} & $5.9841$  & $5.910$($-1.23$\%)& $5.902$($-1.37$\%)& $6.007$($0.38$\%)& $6.007$($0.39$\%)\\ \hline 
\textbf{10} & $3.4353$  & $3.416$($-0.57$\%)& $3.408$($-0.78$\%)& $3.462$($0.77$\%)& $3.463$($0.80$\%)\\ \hline 
\textbf{11} & $3.5810$  & $3.552$($-0.81$\%)& $3.544$($-1.04$\%)& $3.608$($0.75$\%)& $3.609$($0.79$\%)\\ \hline

\end{tabular}
\end{center}
\end{small}
\end{table}

\subsection{Algorithm Performance for Noise-Free Data}\label{sec:resultsSc0}
We contrast the performance of Algorithm~\ref{alg:Biascor_Guo} with Logan-GA, MA1 and KA for noise-free data. The use of a long scan duration (up to $90$ minutes) is to assure that equilibrium is achieved as needed for  GA methods.  For a method for which the bias due to model error is not impacted by the need for equilibrium, a  shorter scan duration is preferred. For the results presented in Table~\ref{tab:Sc0} the DV is calculated for the noise-free case over a  scan duration of just $70$ minutes with $t^*=35$ minutes. Accurate integrals are used so as  to focus the conclusions on the impact of the  model error. 

The KA solutions were obtained using two different optimization algorithms for the solution of the highly nonlinear problem, the interior point   and the Marquardt-Levenberg methods, Matlab{\textregistered} functions \texttt{fmincon} and \texttt{lsqnonlin}, resp. In order to provide the most fair comparison the results presented are for \texttt{fmincon}, which gave the better solutions. The KA solution is very dependent on provision of a good initial value. If the initial values of $k_3$ and $k_4$ are taken very close to their true values,  the estimate of the DV may be nearly perfect. Here we use  initial values for $K_1$, $k_2$, $k_3$ and $k_4$ set to $[0.2$, $0.1$, $0.01$, $0.001]$.  

For Logan-GA and MA1, solutions were also calculated for the scan duration of $T=90$ minutes with $t^*=35$ minutes as illustrated in Table~\ref{tab:overest}. The KA results, not given, which do not require the attainment of equilibrium were comparable for both scan durations as expected. This independence with respect to the requirement of attainment of equilibrium was also observed for  Algorithm~\ref{alg:Biascor_Guo} except for ROI~\textbf{6}. In this case the neglected  part in model (\ref{eq:root1}) is relatively large as compared to that for the other ROIs, i.e. the ratio of $e^{-\alpha_2t}\otimes \Cp(t)$ to $e^{-\alpha_1t}\otimes \Cp(t)$ for ROI~\textbf{6} is greater than that for the other ROIs. A significant reduction in the bias for ROI~\textbf{6} from $-12.71\%$ ($70 $ min.)  to  $-7.39\%$ ($90$ min.) was observed. It is clear, by comparing the results with those in Table~\ref{tab:overest}, that Algorithm~\ref{alg:Biascor_Guo} for a scan duration of just $70$ minutes is much more accurate for the calculation of the DV  than are  Logan-GA and MA1 using scan durations of $90$  minutes.

%\newpage
\begin{table}[bhtp]
\begin{small}
\begin{center}
\caption{DV calculated by Logan-GA, MA1, KA and Algorithm~\ref{alg:Biascor_Guo} for a $70 $ minutes scan duration with $t^*=35$ minutes.  In each case the percentage bias is listed in parentheses. \label{tab:Sc0}}
\begin{tabular}{|c|c|c|c|c|}
\hline
ROI & Logan-GA & MA1 & KA & Algorithm~\ref{alg:Biascor_Guo}\\\hline
\textbf{1} & $3.395$($-14.54$\%)& $3.392$($-14.61$\%)& $3.928$($-1.12$\%)& $4.014$($1.05$\%) \\ \hline 
\textbf{2} & $6.511$($-2.64$\%)& $6.506$($-2.71$\%)& $6.552$($-2.02$\%)& $6.777$($1.34$\%) \\ \hline 
\textbf{3} & $3.436$($-14.65$\%)& $3.433$($-14.71$\%)& $3.982$($-1.08$\%)& $4.066$($1.00$\%) \\ \hline 
\textbf{4} & $8.163$($-5.86$\%)& $8.157$($-5.92$\%)& $8.535$($-1.56$\%)& $8.743$($0.83$\%) \\ \hline 
\textbf{5} & $7.931$($-6.88$\%)& $7.925$($-6.95$\%)& $8.383$($-1.57$\%)& $8.530$($0.16$\%) \\ \hline 
\textbf{6} & $3.004$($-33.97$\%)& $3.007$($-33.92$\%)& $4.675$($2.74$\%)& $3.972$($-12.71$\%) \\ \hline 
\textbf{7} & $3.293$($0.63$\%)& $3.292$($0.58$\%)& $3.188$($-2.59$\%)& $3.277$($0.12$\%) \\ \hline 
\textbf{8} & $3.555$($-5.15$\%)& $3.549$($-5.30$\%)& $3.679$($-1.84$\%)& $3.784$($0.95$\%) \\ \hline 
\textbf{9} & $5.847$($-2.28$\%)& $5.842$($-2.37$\%)& $5.859$($-2.10$\%)& $6.008$($0.40$\%) \\ \hline 
\textbf{10} & $3.376$($-1.73$\%)& $3.371$($-1.87$\%)& $3.361$($-2.17$\%)& $3.451$($0.47$\%) \\ \hline 
\textbf{11} & $3.506$($-2.09$\%)& $3.501$($-2.24$\%)& $3.505$($-2.11$\%)& $3.585$($0.10$\%) \\ \hline 
\end{tabular}
\end{center}
\end{small}
\end{table}
%\newpage

In contrasting the results with respect to only the bias in the calculation of the DV it is clear that Algorithm~\ref{alg:Biascor_Guo}  leads to significantly more robust solutions than Logan-GA1 and MA1   for noise-free data. On the other hand,  the KA approach can lead to very good solutions, comparable and perhaps marginally better than  Algorithm~\ref{alg:Biascor_Guo}.  For ROI~\textbf{6}, for which the KA solution is significantly better, we recall that the solution depends on the initial values of the parameters. Changing the  initial $k_6$ to $0.01$, %not an unrealistic guess \Remark{No, Actually, 0.01 is a quite realistic initial, see table 1 for the rate constants},, 
the resulting bias in the DV of ROI~\textbf{6} calculated by KA is increased to $31.75\%$. On the other hand,  %demonstrating the dependence of KA on the initial value, which is not required  
Algorithm~\ref{alg:Biascor_Guo} is not dependent on specifying initial values, and is thus more computationally robust.  

% We conclude from the simulations with noise-free data, that the use of the new algorithm for finding the solution defined by model (\ref{eq:root1}) using just a $70 $ minutes scan may lead to more accurate estimates of the DV than Logan-GA and MA1 even over a $90$ minutes duration. 

\subsection{Experimental Design for Noisy Data}
While the results with noise-free data support the use of Algorithm~\ref{alg:Biascor_Guo}, it is more critical to assess its performance for noise-contaminated simulations. The experimental evaluation for noisy data is based on the noise-free input $u(t)$ and noise-free output  $\CT(t)$, one output TTAC for each of the eleven  parameter sets given in Table~\ref{tab:ki}. Noise contamination of the input function and these TTACs is obtained as follows.

\subsubsection{The Noise-Contaminated TTAC Data}
For a given noise-free  decay-corrected concentration TTAC, $\CT(t)$, Gaussian   ($G(0,\sigma(\CT(t))$) noise at each time point $t_i$ is modeled using the approach in \cite{logan2001a,varga2002mod, ichise2002str}. The standard deviation in the noise at each time point $t_i$, depends on  the frame time interval $\Delta t_i$ in seconds, the tracer decay constant $\lambda$ ($0.034$ for $^{11}C$) and a scale factor $Sc$
\begin{equation}\label{var_y}
\sigma(\CT(t_i))=Sc\sqrt{ \frac{\CT(t_i) e^{\lambda t_i}}{\Delta t_i}}.
\end{equation}
The resulting coefficients of variation $\mathrm{CV}_{\mathrm{T}}$ (ratio $\sigma(\CT(t_i))$ to %the mean of 
$\CT(t_i)$), for scale factors $1$ and $2$, are illustrated in Figure~\ref{fig:CV}.
\begin{figure}[tbp]
\centerline{\subfigure[]{
\includegraphics[scale=.45]{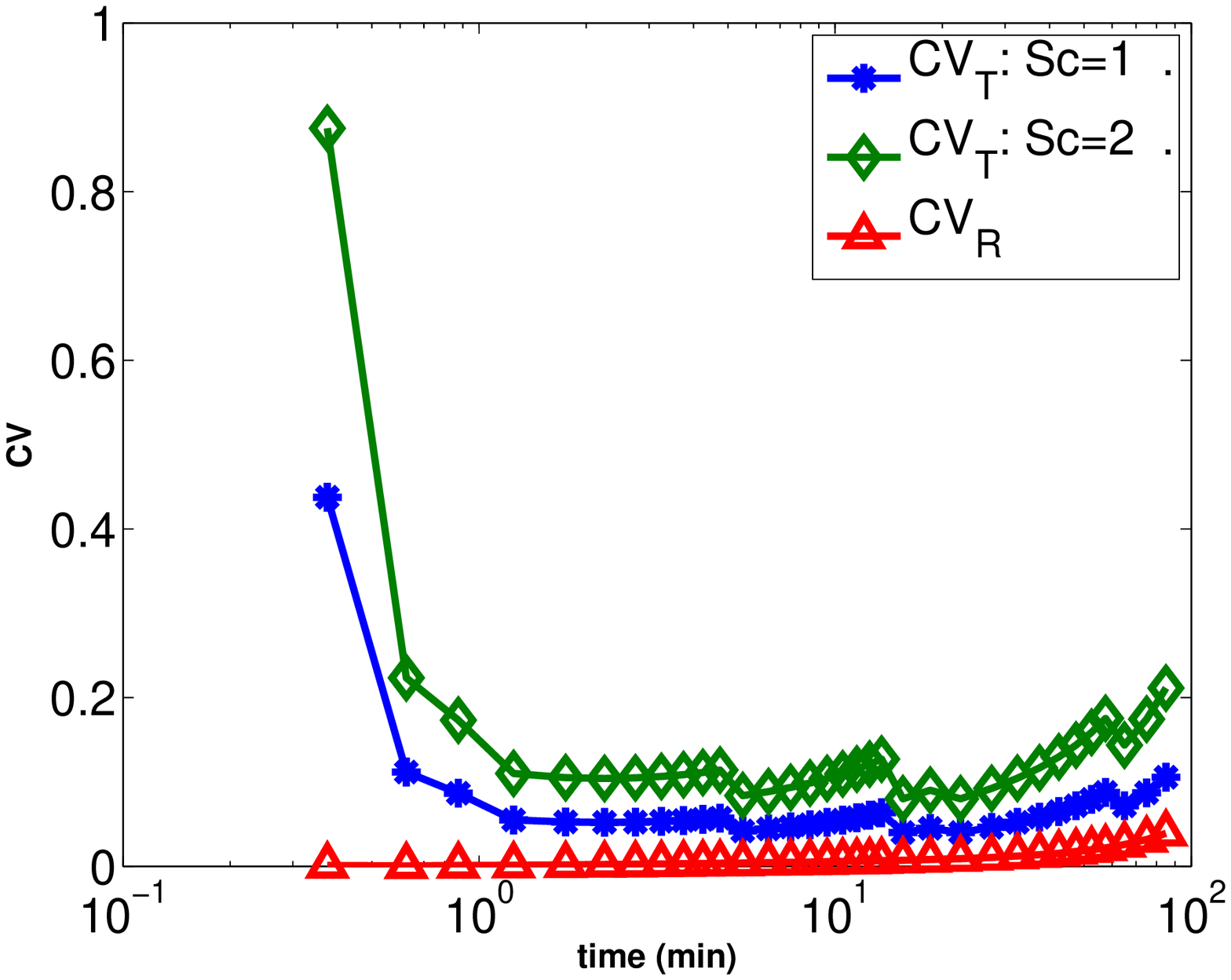}}
\subfigure[]{\includegraphics[scale=.45]{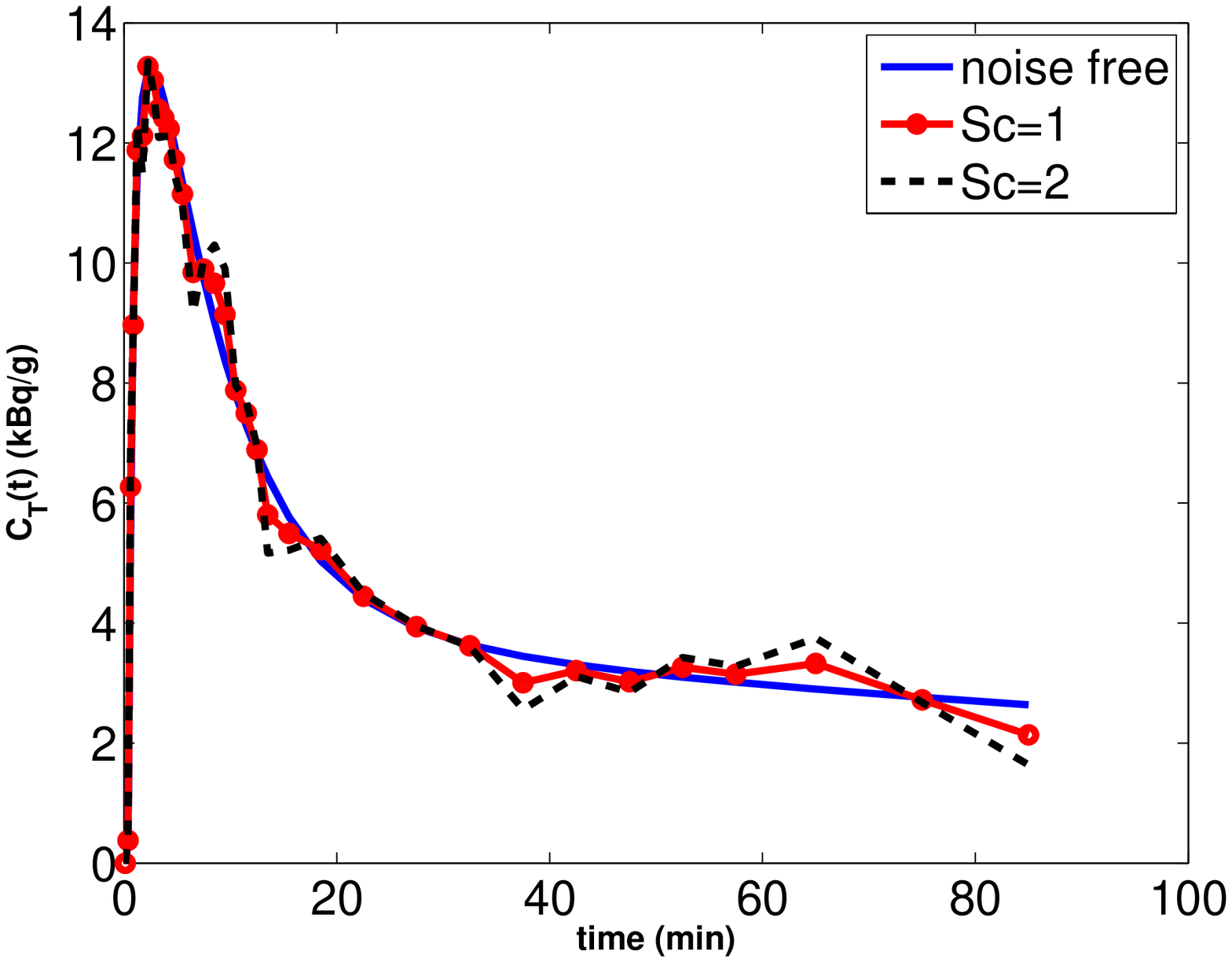}}
}
\caption{(a) The coefficients of variation CV$_{\mathrm{T}}$ for  the noisy TTAC associated with ROI~\textbf{3}, obtained with  $Sc=1$ and $Sc=2$, resp.  and CV$_{\mathrm{R}}$, for the input function calculated for  $e=50\%$, $\mu=0.5$ml and  $\Delta w_i=100$ seconds. (b) The noise-free and noisy TTACs for ROI~\textbf{3} obtained with $Sc=1$ and $Sc=2$, resp.}
\label{fig:CV}
\end{figure}

\subsubsection{The Noise-Contaminated Input Function}
The noise in the input function can be attributed to two sources, system and random noise. Although the random $\gamma$-ray emission follows a Poisson distribution, we  use the limiting result that a large mean Poisson distribution is approximately  Gaussian to model this randomness as Gaussian. Thus both sources are modeled as Gaussian but with different variance. Consider first the following model for determining the randomness of the $\gamma-$ray emissions. Suppose a $\mu$ ml blood sample is placed  in a $\gamma$-ray well counter which has efficiency $e$ and the measured counts over $\Delta w_i$ seconds are $n(t_i)$.  Then the measured decay corrected concentration (kBq/ml) is 
$$\Cp(t_i)=\frac{n(t_i)e^{\lambda t_i}}{1000\Delta w_i \mu e},$$
where $1000$ is a normalization factor to convert the counts to ``kilo'' counts. %\Remark{I change kBq/ml to kilo. if you convert counts to kBq/ml, it involves the 1000, volume, and counting time. so divided by 1000 only for kilo counts}. 
Then, assuming that the mean of $\Cp(t_i)$ (or its true value) is  $u(t_i)$ as given in (\ref{inpf.eq}), the 
standard deviation in the measurement of $\Cp(t_i)$ due to random effects  is  $\sigma_{\mathrm{R}}(\Cp(t_i))=\sqrt{u(t_i)e^{\lambda t_i}/(1000\Delta w_i  \mu e)}$. The coefficient of variation, C$V_\mathrm{R}=\sigma_{\mathrm{R}}(\Cp(t_i))/u(t_i)$, which results from this  random noise is shown in Figure~\ref{fig:CV}. It  is assumed in the experiments that each blood sample has volume $\mu=0.5ml$, the count duration is $\Delta w_i=100$ seconds and the well counter efficiency is $e=50\%$. Then, denoting the coefficient of variation due to system noise by C$V_\mathrm{S}$, the noise-contaminated input is given by
\begin{equation}\label{eq:noisyinput}
 \Cp(t_i)=u(t_i)(1+(\mathrm{CV}_\mathrm{R} + \mathrm{CV}_\mathrm{S})\eta_i),
\end{equation}
where $\eta_i$ is selected from a standard normal distribution (G$(0,1)$), and in the simulations we use C$V_\mathrm{S}=0.05$, see Figure \ref{fig:inpf}.

\subsection{Experimental Results for Noisy Data}\label{subsec:noisy}
Two hundred random noise realizations are generated for each input-TTAC pair, and for each noise level ($Sc=1$, $2$).  The distribution volume is calculated for each experimental pair using Logan-GA, MA1, KA and Algorithm~\ref{alg:Biascor_Guo}. In each case two scan durations are  considered, $70 $ and $90$  minutes respectively, and $t^*=35$ minutes. Unlike the noise-free case, the numerical quadrature for $\int_0^t\Cp(\tau){\rm d}\tau$  uses only the samples at scan points $\Cp(t_i)$. %The numerical results reported in Tables~\ref{tab:Sc1} and \ref{tab:Sc2}, noise levels $Sc=1$ and $Sc=2$, resp., are the mean and standard deviation of the bias normalized to the true DV, given as a percentage of the true DV.  

We present histograms for the  percentage relative error of the bias $100(\mathrm{DV}_{\mathrm{est}}-\DVT)/\DVT$ in order to provide a comprehensive contrast of the methods.  Figure~\ref{fig:DV-histall} shows the histograms for all eleven ROIs, with the range of the error for each method indicated in the legend. The figures (a)-(b) are for scan windows of $90$ minutes, for noise scale factors $Sc=1$ and $Sc=2$ while  (c)-(d) are for  scan windows of $70$ minutes.  Figure~\ref{fig:DV-hist3} provides equivalent information for a representative cortical region ROI~\textbf{3}. It is clear that the distributions of the relative errors for KA and MA1 are far from normal; KA has a significant positive tail while  Logan-GA has strong negative bias.  MA1 has   unacceptably long tails except for the case of low noise with long scan duration, i.e. $Sc=1$ with $90$ minutes scan duration. On the other hand, the histogram for Algorithm~\ref{alg:Biascor_Guo} is close to a Gaussian random distribution; the mean is near zero and the distribution is approximately symmetric. Moreover,  Algorithm~\ref{alg:Biascor_Guo}  performs well, and is only outperformed marginally by MA1  for the lower noise and longer time window case. On the other hand, there are some situations, particularly for MA1,  in which the relative error is less than $-100\%$; in other words, the calculated DVs are negative. Such \textit{unsuccessful} results occur only for the higher noise level ($Sc=2$). While there was only one such occurrence for the Logan-GA ($70$ min. with ROI~\textbf{9}) , there were $40$ such occurrences for MA1, $33$ for the shorter time interval of  $70$ minutes (ROIs~\textbf{1}, \textbf{3}, \textbf{4}, \textbf{5}, \textbf{6}, \textbf{8} and \textbf{9}) and $7$ for the longer interval of $90$ minutes, (ROIs~\textbf{1} and  \textbf{6}).  The reason for the negative DV for MA1 is discussed in Section~\ref{subsec:negDV}. From the results for the higher noise $Sc=2$ we conclude that  Algorithm~\ref{alg:Biascor_Guo} using the shorter $70$ minutes scan duration outperforms the other algorithms, even in comparison to their results for the longer scan duration.  

\begin{figure}[htbp]
\centerline{\subfigure[]{\includegraphics[scale=.45]{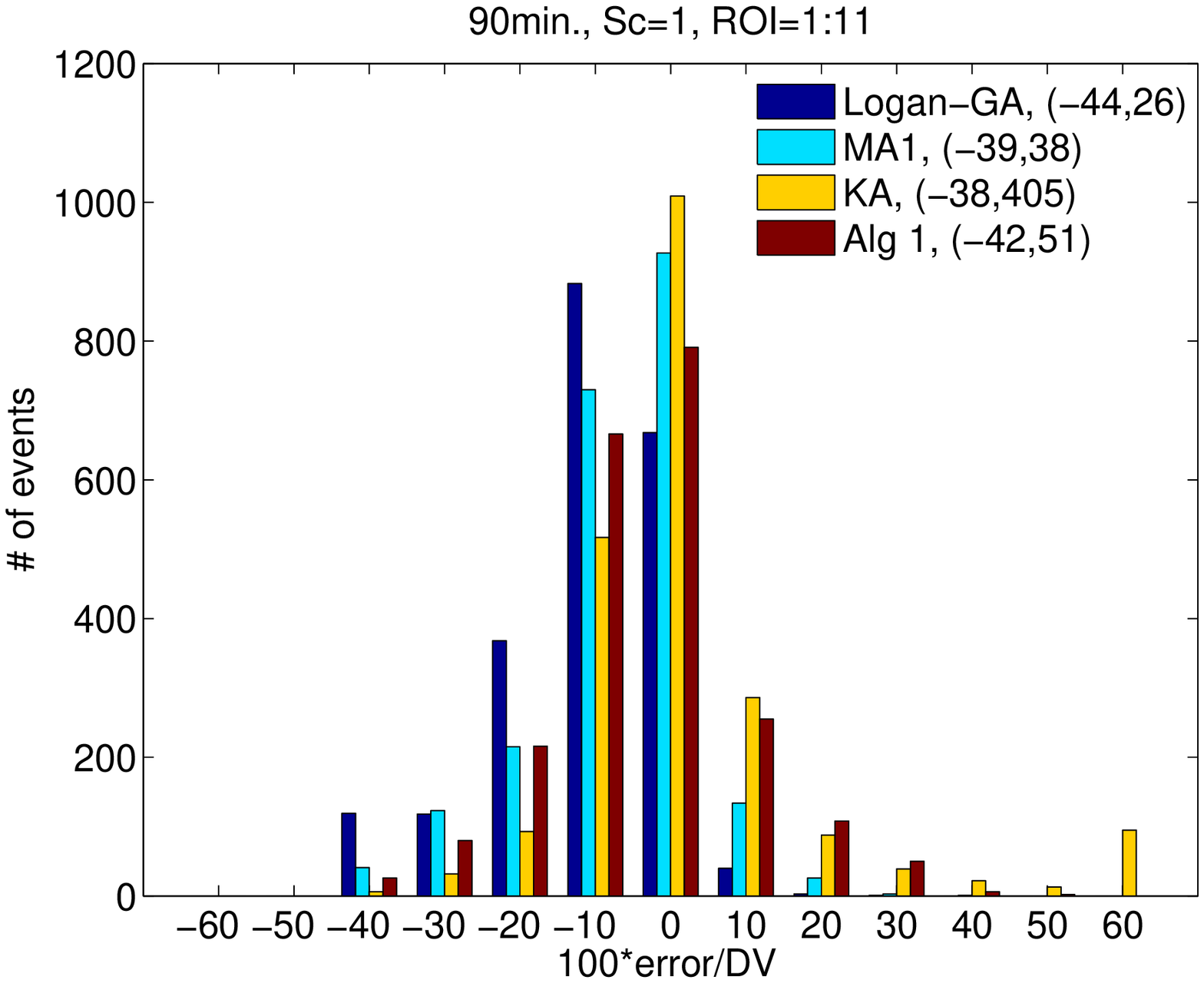}}\subfigure[]{
\includegraphics[scale=.45]{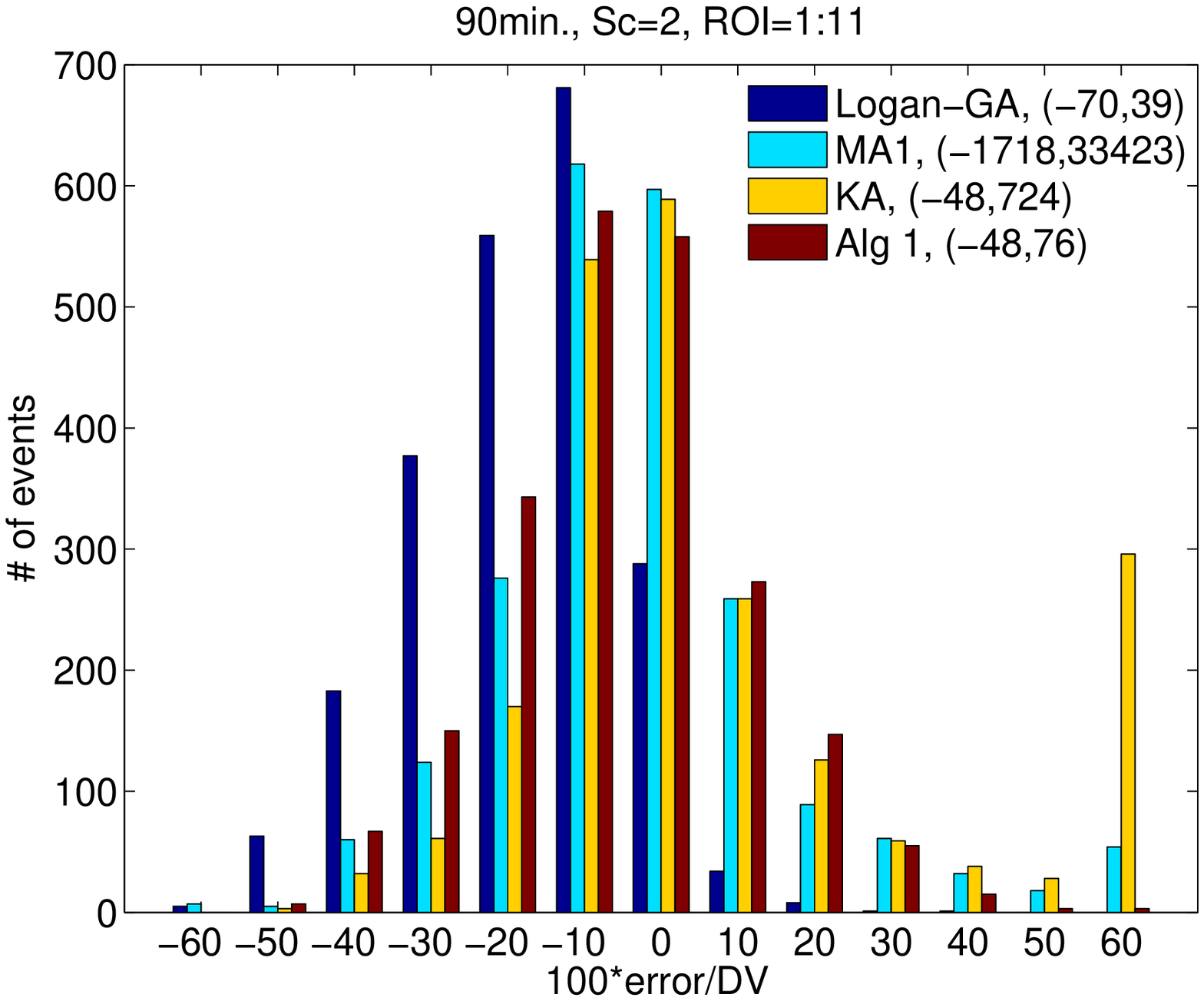}}}
\centerline{\subfigure[]{\includegraphics[scale=.45]{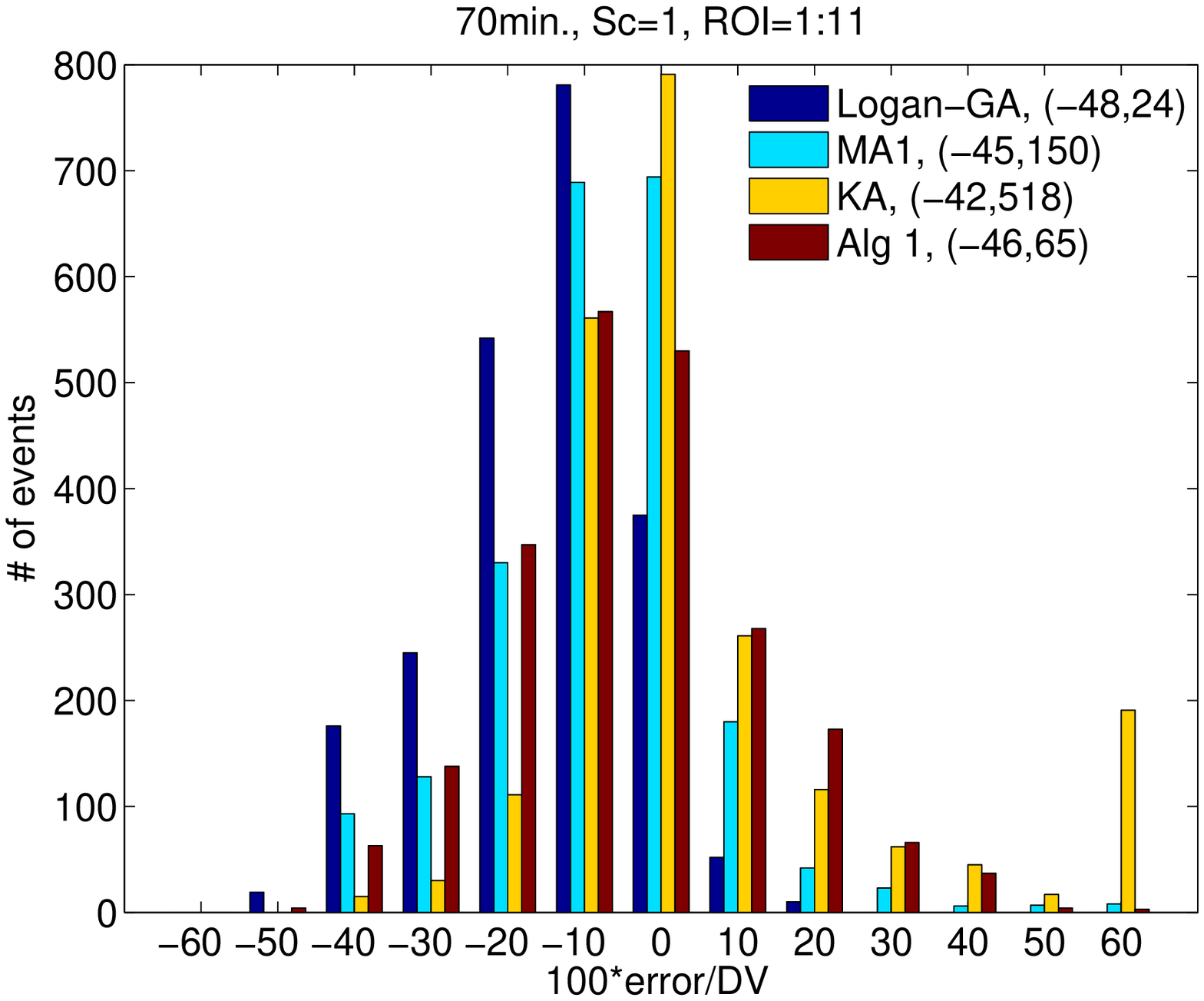}}\subfigure[]{
\includegraphics[scale=.45]{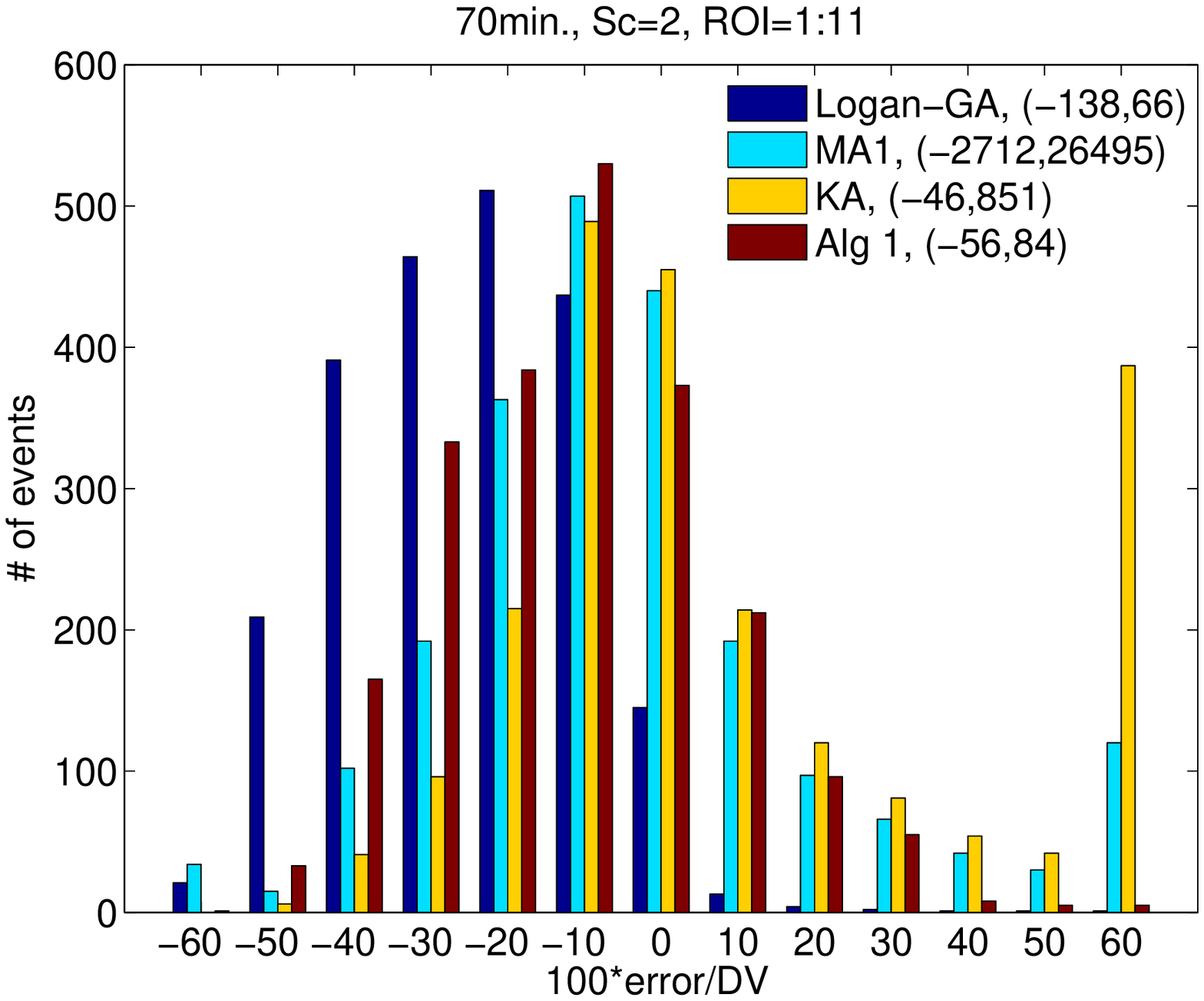}}}
\caption{Histograms for normalized error (in percentage),$100(\mathrm{DV}_{\mathrm{est}}-\DVT)/\DVT$, of the results for all eleven ROIs and four methods. The error ranges are presented in the legends.}
\label{fig:DV-histall}
\end{figure}

\begin{figure}[htbp]
\centerline{\subfigure[]{\includegraphics[scale=.45]{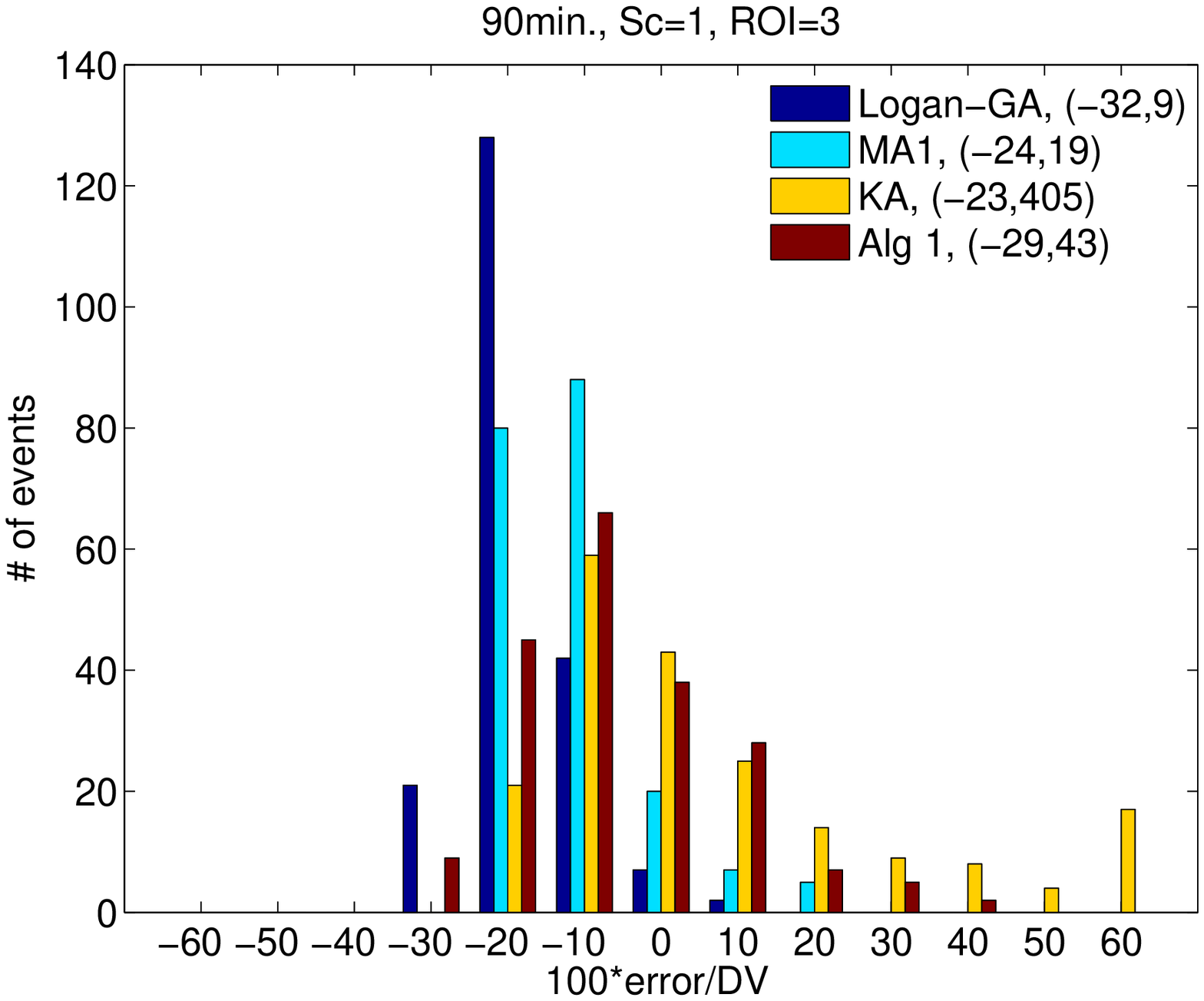}}
\subfigure[]{\includegraphics[scale=.45]{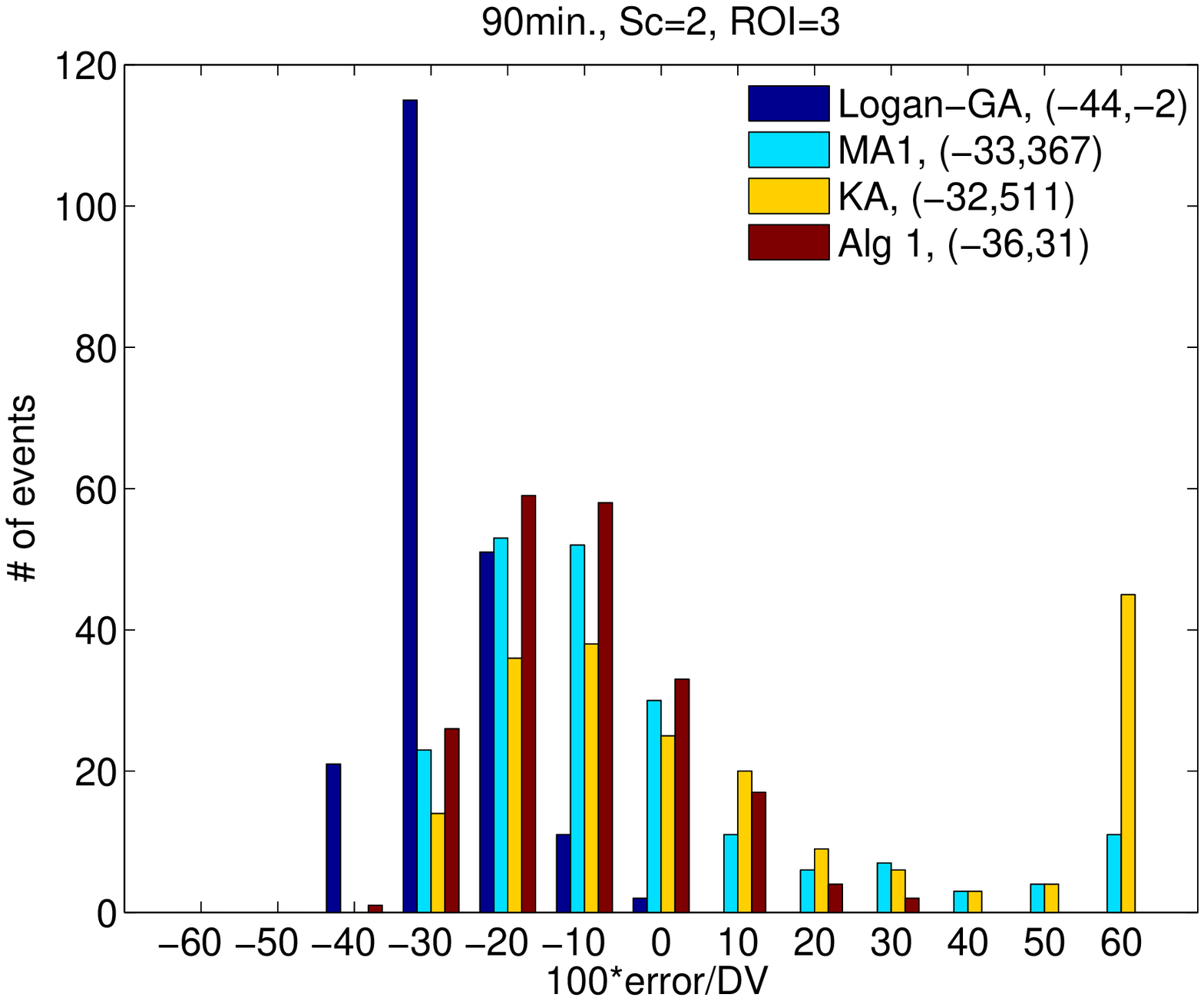}}}
\centerline{\subfigure[]{\includegraphics[scale=.45]{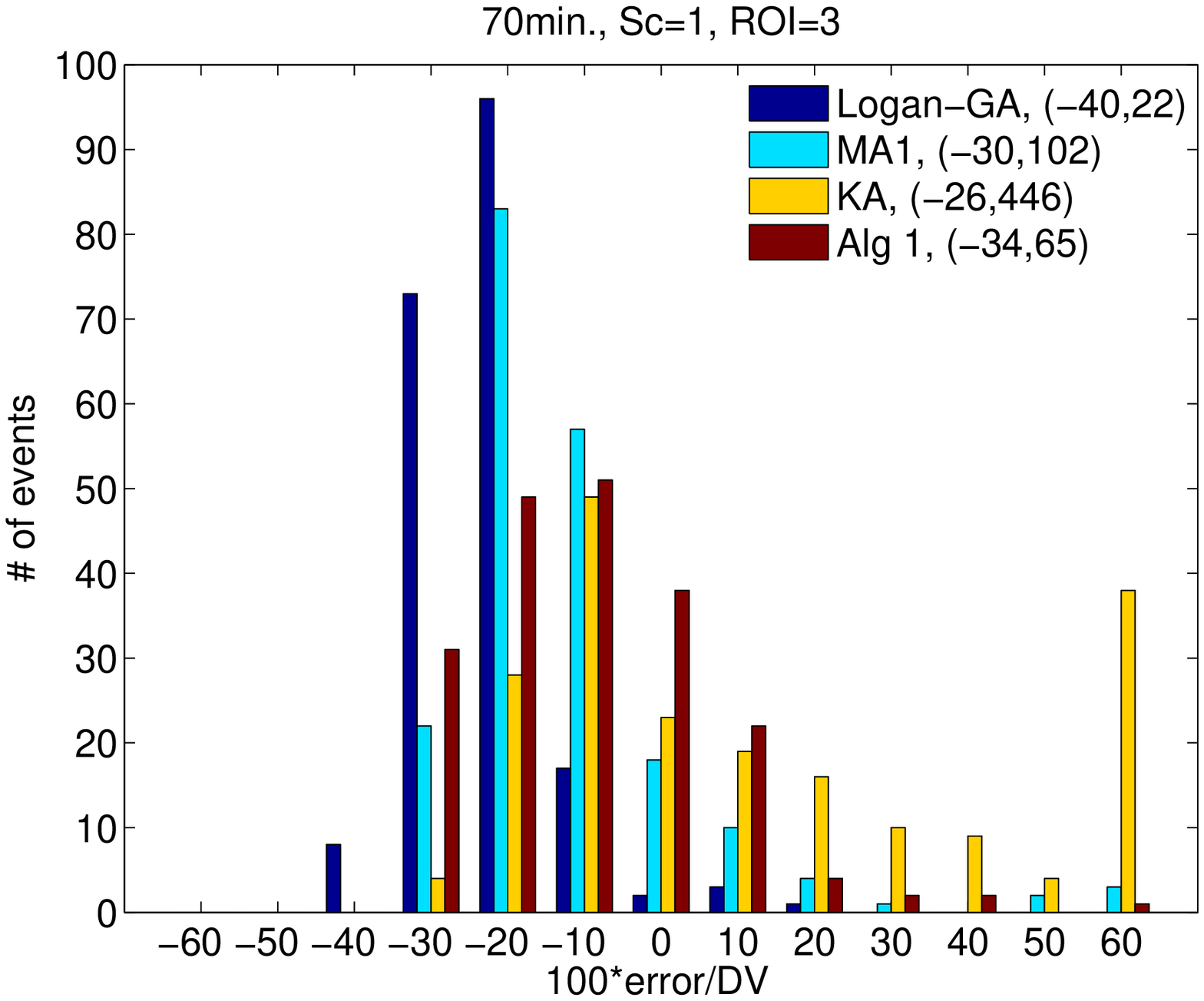}}
\subfigure[]{\includegraphics[scale=.45]{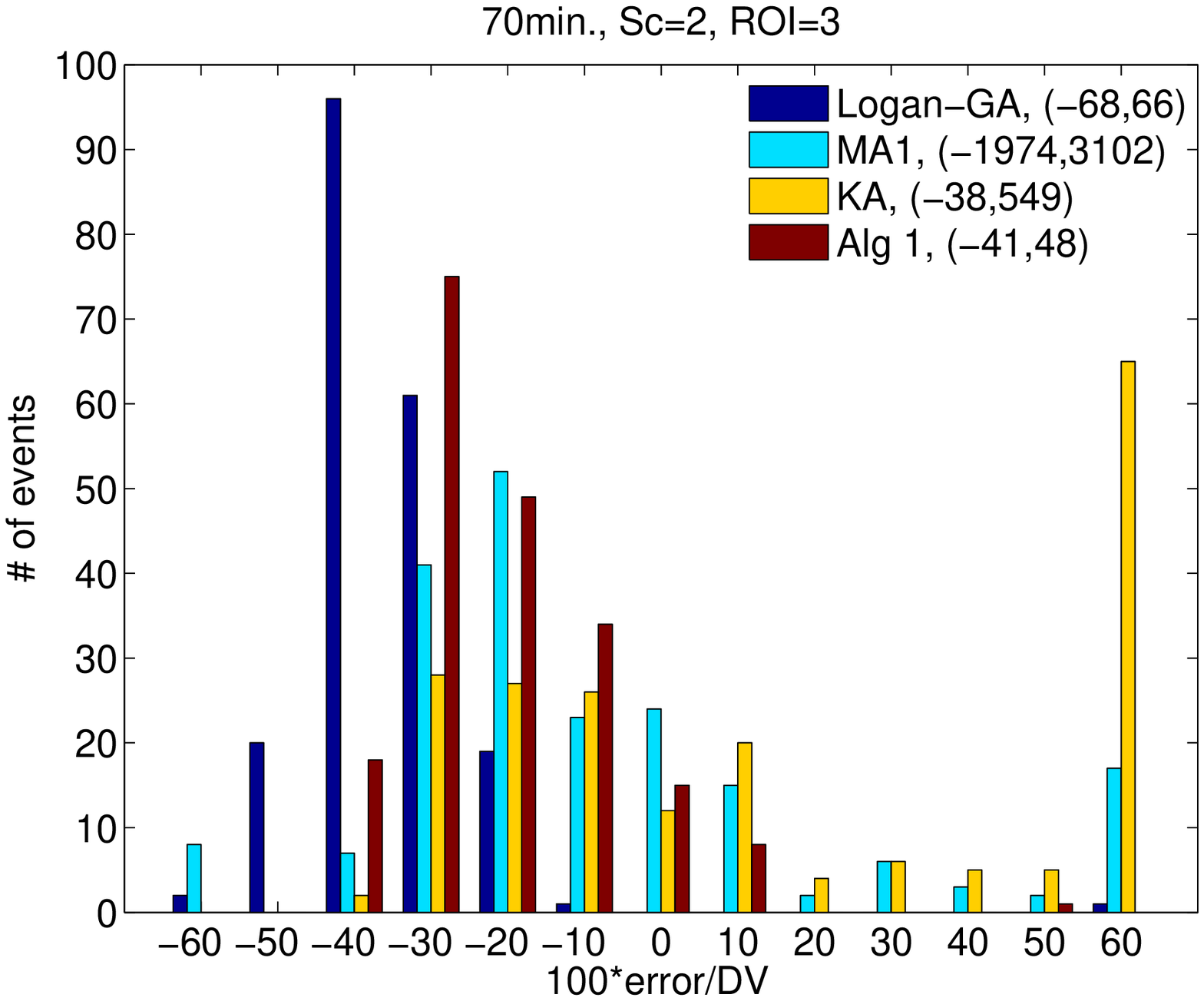}}}
\caption{Histograms for normalized error (in percentage),$100(\mathrm{DV}_{\mathrm{est}}-\DVT)/\DVT$, of the results for  ROI~\textbf{3} and four methods. The error ranges are presented in the legends.}
\label{fig:DV-hist3}
\end{figure}

Obviously Algorithm~\ref{alg:Biascor_Guo} is more expensive computationally than Logan-GA and MA1.  In the simulations, the average CPU time, in seconds, per TTAC  was  $0.00083$, $0.00057$, $12.2$ and $0.0036$, for Logan-GA, MA1,  KA and Algorithm \ref{alg:Biascor_Guo},  respectively. The high cost of the KA results from the requirement to use a nonlinear algorithm. Because the KA requires a good initial estimate for the parameters the cost is variable for each TTAC; it is dependent on whether the supplied initial value is a good initial estimate. Indeed the KA results take from $8$ to $25$ seconds, while the costs using the other methods are virtually TTAC independent. %We conclude from the improved estimates of the DV, that the roughly $50\%$ extra cost of the new algorithm is well-justified.

% the condition number is change dramatically. in our simulations,  GA 1e+3--1e+4, MA1--100, MA2 1e+4-1e+5,TLS1 100-200, TLS2 100-300. \\
% 
% \begin{figure}[tbp]
% \includegraphics[scale=0.7]{DV-sc1.eps}\\
% \includegraphics[scale=0.7]{DV-sc2.eps}
% \caption{The distribution of calculated DV by four methods for the eleven  test ROIs with $90$ minutes duration. Two noise cases, $sc=1$ and $sc=2$, are tested by 200 realizations for each case. In each subfigure 1400 DVs, 200 for each ROI, are plotted.  }
% \label{fig:DV-spread}
% \end{figure}

\section{Discussion}\label{sec:diss}
\subsection{Equilibrium Behavior and Dependence on the Size of $k_4$}\label{subsec:equilibrium}
The graphical analysis methods of Logan-type rely on the assumption that the ratio $\CS(t)$ to $\CT(t)$ is approximately constant within a  chosen window $[t^*, T]$. This ratio is plotted against time for  the simulated data for ROIs \textbf{1}  to \textbf{11} in Figure~\ref{fig:CS2CT90}. It is clear that the ratios for ROIs~\textbf{1}, \textbf{3} and \textbf{6} have not reached equilibrium even by $90$ minutes. These are the three data sets with the largest bias reported in Section~\ref{sec:overest} and with smallest $k_4$ (resp. $k_6$). It is certain that  equilibrium is eventually reached. These curves first increase to a peak  at about $120$ minutes for ROIs~\textbf{1} and \textbf{3} and at about $180$ minutes for ROI~\textbf{6} and then decrease before reaching approximately constant values (Figure~\ref{fig:CS2CT720}). On the other hand, increasing the scan duration to more than two hours is not practical.  Moreover, as illustrated in Figure~\ref{fig:GA-plot}, using the linearity of   $\int_0^t \CT(\tau){\it d\tau}/\CT(t)$ versus  $\int_0^t\CT(\tau){\it d\tau}/\Cp(t)$ to verify whether equilibrium has been reached may be misleading. For example, it would appear that all eleven  data sets have achieved equilibrium after roughly $35$ minutes. The arrow in Figure~\ref{fig:GA-plot} points to the marker corresponding to the data calculated at the middle point of the frame from $35$ to $40$ minutes.

\begin{figure}[tbp]
\centerline{
\subfigure[]{\label{fig:CS2CT90}        \includegraphics[height=2.5in,width=2.5in]{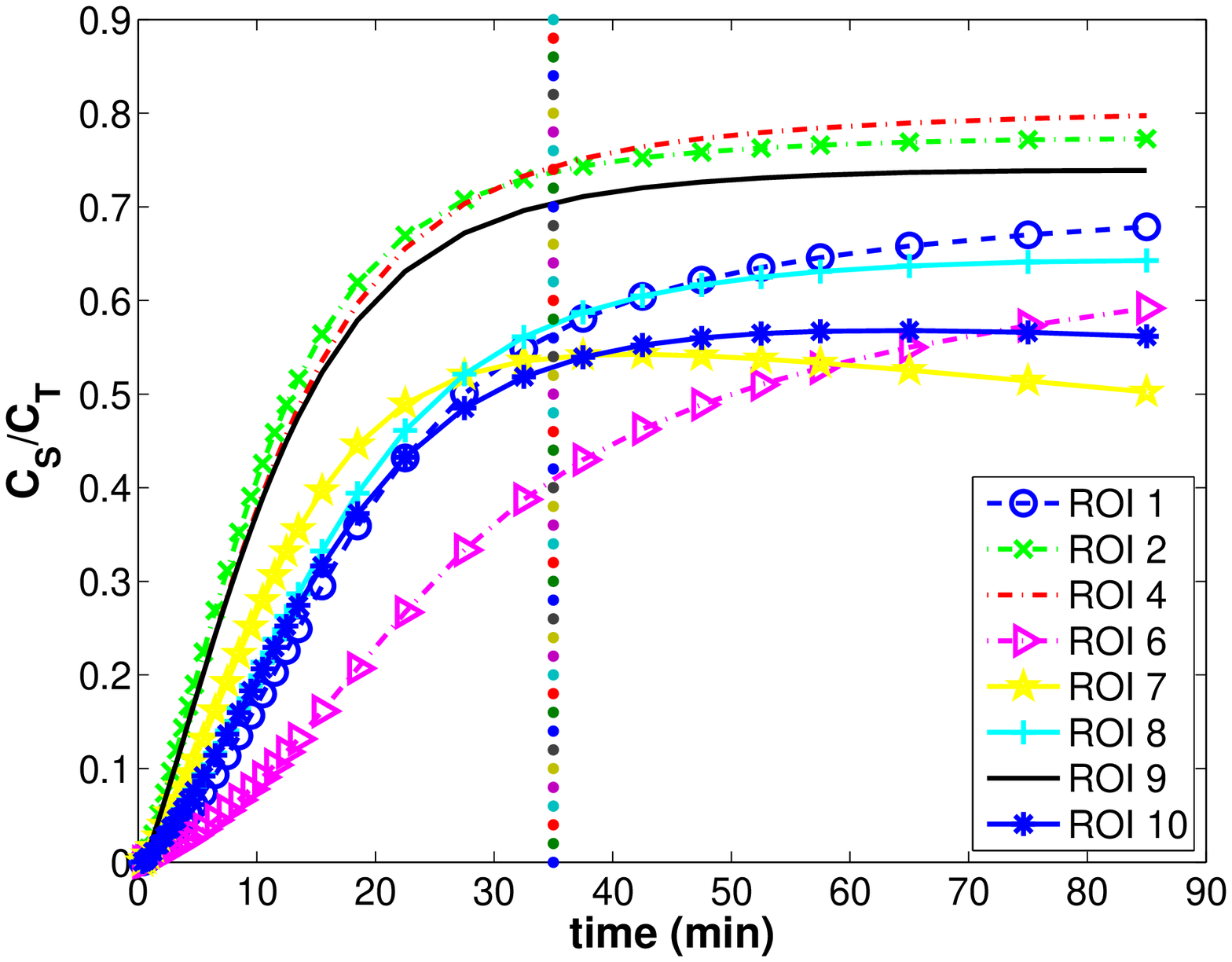}}
     \hspace{.3in}
\subfigure[]{\label{fig:CS2CT720}          \includegraphics[height=2.5in,width=2.5in]{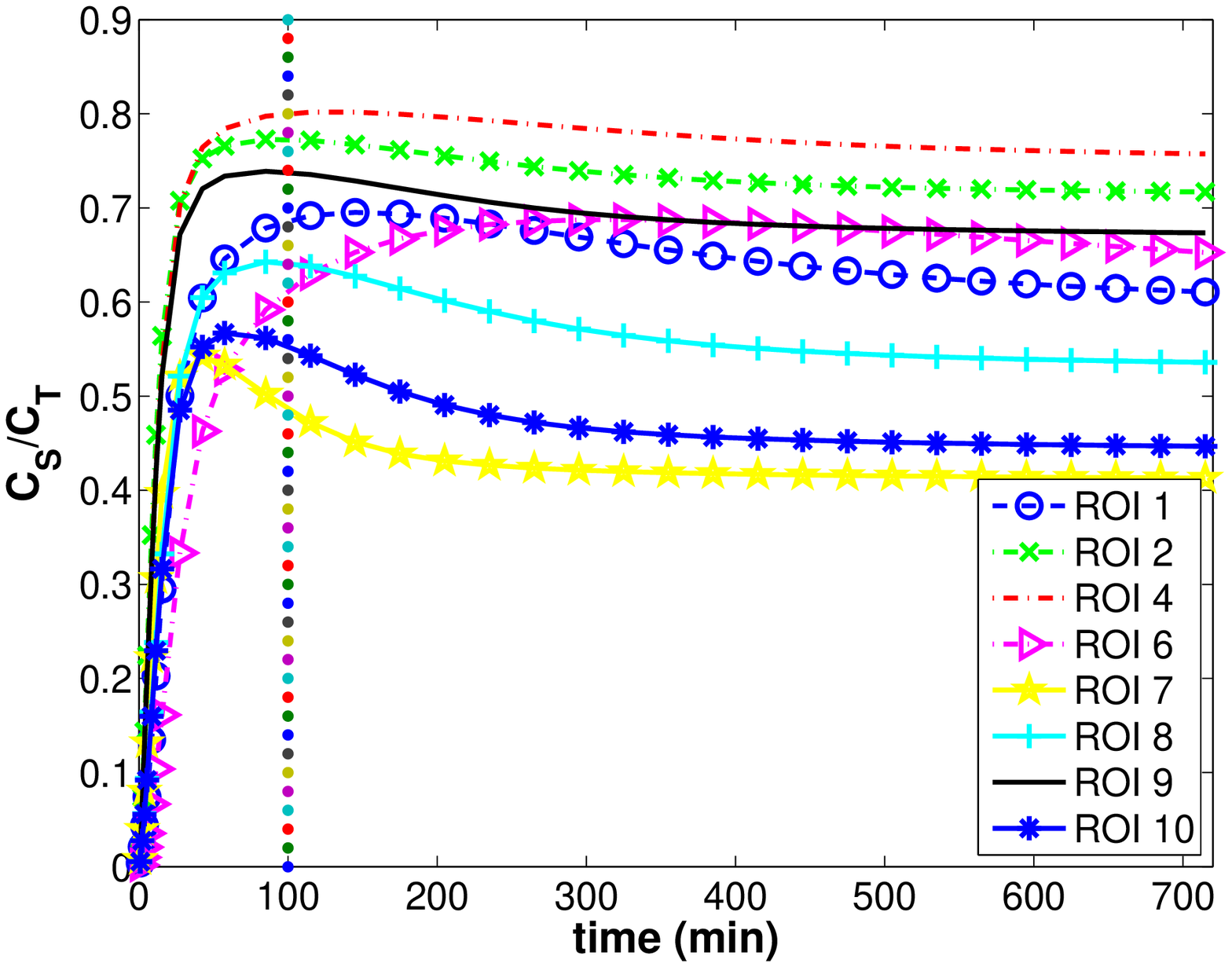}}
}
\caption{$\CS(t)/\CT(t)$ against time for all test ROIs except ROIs~\textbf{3}, \textbf{5} and \textbf{11} for the first $90$ minutes (a) and $720$ minutes (b). Dotted vertical lines are plotted at time $t^*=35$ minutes  (a) and $t^*=100$ minutes (b). The curves for ROIs~\textbf{3}, \textbf{5} and \textbf{11} are similar to those for ROIs~\textbf{1}, \textbf{4} and \textbf{10} resp..} 
\label{fig:CS2CT}
\end{figure}
 
\begin{figure}[tbp]
\centerline{
\includegraphics[scale=.45]{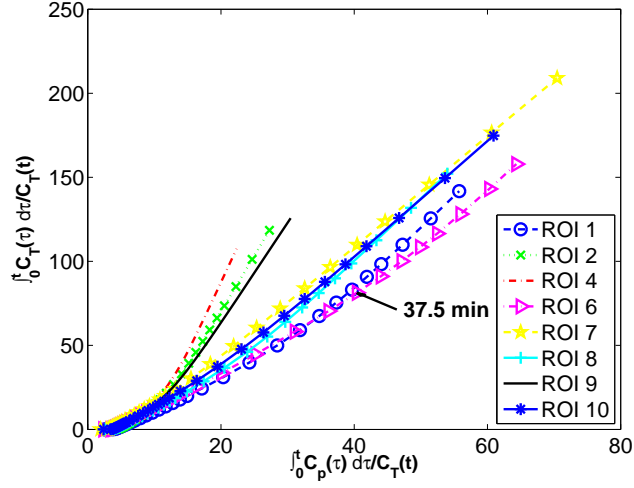}}
\caption{$\int_0^t \CT(\tau){\it d\tau}/\CT(t)$ (y-axis) against $\int_0^t\Cp(\tau){\it d\tau}/\CT(t)$ (x-axis) for  all test ROIs except ROIs~\textbf{3}, \textbf{5} and \textbf{11}  for the first $90$ minutes. The last eight points correspond to the time interval $35$ to $90$ minutes. The curves for ROIs~\textbf{3}, \textbf{5} and \textbf{11} are similar to those for ROIs~\textbf{1}, \textbf{4} and \textbf{10} resp.. The arrow points to the first frame falling in this interval for ROI~\textbf{6}. 
}
\label{fig:GA-plot}
\end{figure}

We illustrate the relation between the bias in the estimate of DV calculated by Logan-GA and $k_4$ in  Figure \ref{fig:Bias_k4}.  As discussed in Section~\ref{subsec:errana}, a small value of $k_4$ may cause a large variation in  $\bar{\bfs}(t)$. This graph verifies that the magnitude of the bias decreases as $k_4$ increases, further verifying that large bias in DV may arise purely due to modeling assumptions in the absence of noise in the data.

\begin{figure}[tbp]
\centerline{\includegraphics[scale=.45]{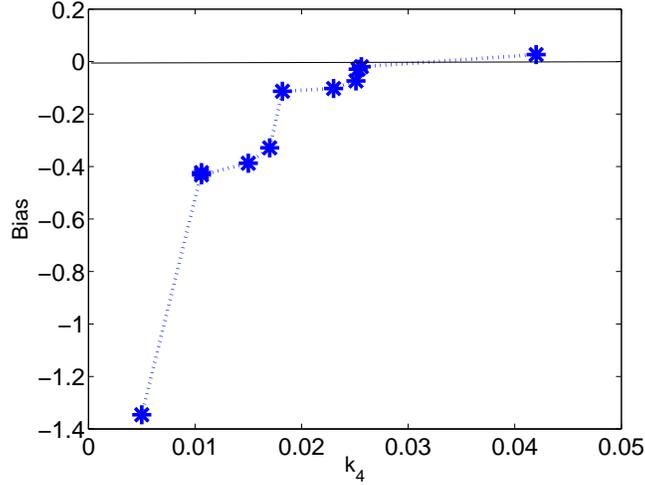}}
\caption{The bias in the Logan-GA estimation of the DV against the value of $k_4$ for the eleven ROIs, assuming noise-free data, a scan duration of $90$ minutes and $t^*=35$ minutes. 
The specific data pairs ($k_4$, bias) are, for ROIs~\textbf{1} to {11}, respectively, 
$(0.0106,   -0.4231)$, $(0.0230,   -0.1024)$, $(0.0106,   -0.4318)$, $(0.0170,   -0.3286)$, 
$(0.0150,   -0.3874)$, $(0.0050,   -1.3459)$, $(0.0420,    0.0267)$, $(0.0182,   -0.1130)$, 
$( 0.0251,   -0.0736)$, $(0.0256,   -0.0195)$, and $(0.0253,   -0.0288)$.}
\label{fig:Bias_k4}
\end{figure}

\subsection{The effects of quadrature error}\label{subsec:quad_effect}
Both Logan-GA and MA1, (\ref{eq:Logan}) and (\ref{eq:Multi-lin1}) resp., require the calculation of integrals $\int_0^t \CT(\tau){\rm d}\tau$ and $\int_0^t\Cp(\tau){\rm d}\tau$. Assume the noise-free measurements $\CT(t_i)$ are derived from the integral over the $i$th frame duration. Thus we can easily recover its integral without introducing error while  quadrature error for calculation of $\int_0^t\Cp(\tau){\rm d}\tau$  due to using a limited number of  plasma samples is unavoidable. The accuracy of the numerical  quadrature   impacts the accuracy of the parameter estimates. Note that we classify the noise effects as another source of bias in DV. 

We recalculate the DV for the experiments reported in Section~\ref{sec:overest}, but now using numerical quadrature for calculation of $\int_0^t\Cp(\tau){\rm d}\tau$ with data sampled one time point per time frame.  The bias for each ROI of the estimated DV using $90$ minutes scan data with $t^*=35$  minutes   is   $-11.83 \% $, $-2.99 \% $, $-11.91 \% $, $-4.88 \% $, $-5.64 \% $, $-30.49 \% $, $-1.22 \% $, $-4.61 \% $, $-2.81 \% $, $-2.40 \% $ and $-2.63 \% $ when calculated using Logan-GA, and  $-12.02 \% $, $-3.10 \% $, $-12.10 \% $, $-5.01 \% $, $-5.77 \% $, $-30.42 \% $, $-1.28 \% $, $-4.87 \% $, $-2.93 \% $, $-2.61 \% $ and $-2.86 \% $ calculated using MA1. It is interesting to note that the DV calculated for ROI~\textbf{7} is no longer an over-estimate. This does not contradict the result of Theorem~\ref{thm:Logan_bias}, which predicts that the DV for ROI~\textbf{7} will be over-estimated due to model error, provided that the other aspects of the calculation are accurate. Now using a less accurate quadrature the negative bias due to quadrature error canceled the positive bias due to the model error. Indeed, for all eleven test cases the impact of the less accurate quadrature is to shift the bias down, i.e. it is more negative as compared to the equivalent more accurate calculations shown in  Table~\ref{tab:overest}. 

%While there is progress in data measurement in the PET community, due to improvements in both hardware and software, the use of shorter time frames has been proposed and practised. % For example, ADNI FDG-PET data acquisition use six $5$ minutes frames instead of a single 30 minute frame. 
%Additional studies are needed to investigate the trade-off between the quadrature accuracy ($\int_0^t \CT(\tau){\rm d}\tau$ and $\int_0^t\Cp(\tau){\rm d}\tau$) and the variability in PET data ($\CT(t_i)$) due to lower count statistics.

\subsection{Bias and classification between AD and NC subjects}
In the eleven simulated ROIs, large under-estimation of the DV calculated by Logan-GA and MA1 is observed for ROIs~\textbf{1} (NC Cort), \textbf{3} (NC PCG) and \textbf{6} (AD Cere). A lower value of the  DV in the  cortical regions of NCs and in the cerebellum for AD subjects  will result in under-estimation of the DVR for NCs and  over-estimation of the   DVR for AD subjects when the cerebellum is used as the reference region for the DVR calculation. Thus, the difference between AD and NC can be artificially enhanced, and viewed as a positive outcome associated with the bias of Logan-GA and MA1. This conclusion, however, can not be generalized. It is unknown whether it is always the case that AD/NC have small/large $k_6$ in cerebellar regions and relatively large/small $k_4$ in cortical regions. Confirmation of these assertions would suggest, based on the discussion in Sections~\ref{subsec:errana} and \ref{subsec:equilibrium},  that the DVR is over-estimated for AD subjects and under-estimated for healthy subjects (also see Figure~\ref{fig:Bias_k4}). In addition, more subtle differences, such as the ones between mild cognitive impairment (MCI) and NC, or among NC with differential genetic risk for AD, may make the effects of bias much less predictable. %Overall, our goal is to estimate DV as accurate/reliable rather than settling down on the simplest difference between AD patients and normal controls.
Consequently, we evaluate the quantification methods based on their bias because the goal of these methods is to estimate the DV as accurately as possible.

% While, this would be useful for distinguishing AD and NC subjects, it is not immediate that we can conclude from this small number of ROI tests that Logan-GA and MA1 produce good DV values.  But if we could answer the following questions with confidence,  we could conclude that Logan-GA and MA1 generate good values for the DV:
% \begin{enumerate}
% \item Do AD patients have a small $k_6$  in the cerebellum and a relatively large $k_4$ in the cortex?
% \item Do normal healthy subjects have a large $k_6$  in the cerebellum and a relatively small $k_4$ in the cortex?
% \end{enumerate}
% Confirmation of these statements would suggest, based on the discussion in Sections~\ref{subsec:errana} and \ref{subsec:equilibrium},  that the DVR is over-estimated for AD subjects and under-estimated for healthy subjects (also see Figure~\ref{fig:Bias_k4}).  Although the limited data sets used for our simulations suggest that AD subjects have higher $k_4$ than NC subjects for cortical regions, additional data needs to be made available to verify the assertion. Consequently, we evaluate the quantification methods based on their bias because the goal of these methods is to estimate the DV as accurately as possible. 

\subsection{When does MA1 fail?} \label{subsec:negDV}
As noted in Section~\ref{subsec:noisy}, MA1 generates some results with negative DVs. Such results are reported as \textit{unsuccessful}  in Ichise's original paper \cite{ichise2002str}. Careful study of these results shows that the negative DVs arise when $-1/b$ has the wrong sign.  For most radioligand binding studies $1/b$ is a small positive number because $b>(k_3+k_4)/(k_2k_4)$, which is usually larger than $10$, see Remark~(\ref{remark:alpha1}) of Algorithm~\ref{alg:Biascor_Guo}. Thus a small error in the estimate of  $-1/b$  due to large noise in the data may change its sign. This  in turn impacts the sign of the estimate of the DV.

\section{Conclusions} \label{sec:conc}
In this article, we quantified the model error in estimating distribution volume using graphical analysis methods. We described the conditions under which the DV is either  over- or under-estimated, and  quantified the bias caused by model error. We validated our findings through simulations with noise-free data. To reduce the impact of model error, we added a simple nonlinear term to the fundamental linear model MA0, and presented a new algorithm for its solution. Simulations with noisy data demonstrate that the new algorithm is cost-effective and robust even for shorter scan durations.  For PIB-PET studies, the new method using shorter scan data ($70 $ minutes) outperforms, or is at least as good as, Logan-GA, MA1 and KA methods using longer scan data ($90$ minutes). %. This result has considerable impact on practical clinical situations, where the reduction of scan time will reduce cost and increase patient comfort. 
The proposed approach can be easily extended for DVR estimation. This is a focus of our future work. 

%Future work will address analysis and methodology adjustments associated with the bias that results from  noise introduced by the PET measuring process,  and by numerical quadrature.  

\section{Acknowledgment}
This work was supported by grants from the state of Arizona (to Drs. Guo, Reiman, Chen and Renaut), the NIH (R01 AG031581, R01 MH057899 and P30 AG19610 to Dr. Reiman) and the NSF (DMS 0652833 to Dr. Renaut and DMS 0513214 to Drs. Renaut and Guo). %, and grants NIA R01 AGR200737,
The authors thank researchers from the University of Pittsburgh for their published findings, including information about PiB input function and rate constants.

\section{Appendix A: fundamental theory for Corollary~\ref{thm:Logan_bias}} \label{sec:app}
Here we present the theoretical result from which Theorem~\ref{thm:Logan_bias} is obtained. We use the notation that $\bfa=(a_1$ , $a_2$, $\cdots$, $a_n)^T$ and $\bfb=(b_1$, $b_2$, $\cdots$, $b_n)^T$, are vectors with entries $a_i$ and $b_i$, resp. The notation $\bfa/\bfb$ and $\bfa\circ \bfb$ denotes component wise division and multiplication, namely entries $a_i/b_i$ and $a_ib_i$,  $\|\bfa\|_1$ is  $\sum_{i=1}^n|a_i|$ and $\|\bfa\|_2 =\sqrt{a_i^2+a_2^2+\cdots+a_n^2}$  is the Euclidean norm. We call  $\bfa$  decreasing (\textbf{increasing})  if  $a_1 \geq a_2 \geq \cdots \geq a_n$ ($\mathbf{a_1 \leq a_2 \leq \cdots \leq a_n}$), and non-constant decreasing (\textbf{non-constant increasing}) if  it is decreasing (\textbf{increasing}) and  at least one of the $\geq$ ($\mathbf{\leq}$) signs is strict, $>$ ($\mathbf{<}$). If all of the $\geq$ ($\mathbf{\leq}$) signs are strict, we 
call  $\bfa$  strictly decreasing (\textbf{strictly increasing}). A vector $\bfa$ is constant if $a_i=a$ for some constant $a$ and for all $i$. 
\begin{lemma} {\rm ( Chebyshev's sum inequality \cite{Gradmathtable} )} Given real numbers \\ $a_1 \geq a_2 \geq \cdots \geq a_n$ and $b_1 \geq b_2 \geq \cdots \geq b_n,$ then
\begin{equation}\label{cheb_ge}
 \frac{1}{ n} \displaystyle\sum_{k=1}^n a_kb_k \geq \left(\frac{1}{ n}\displaystyle\sum_{k=1}^n a_k\right)\left(\frac{1}{ n}\displaystyle\sum_{k=1}^n b_k\right).
\end{equation}
Similarly, if $a_1 \geq a_2 \geq \cdots \geq a_n$ and $b_1 \leq b_2 \leq \cdots \leq b_n,$    then
\begin{equation}\label{cheb_le}
\frac{1}{ n} \displaystyle\sum_{k=1}^n a_kb_k \leq \left(\frac{1}{ n}\displaystyle\sum_{k=1}^n a_k\right)\left(\frac{1}{ n}\displaystyle\sum_{k=1}^n b_k\right).
\end{equation}
\end{lemma}

In the above Chebyshev's sum inequalities the numbers are not required to be positive and  the equality is true if and only if one of the two vectors, $\bfa$ or $\bfb$, is a constant vector.  If $\bfa$ and $\bfb$ are positive vectors, the Chebyshev's sum inequalities can be expressed as $\bfa^T\bfb\ge \frac{1}{n}\|\bfa\|_1\|\bfb\|_1$ and  $\bfa^T\bfb\le \frac{1}{n}\|\bfa\|_1\|\bfb\|_1$.

\begin{lemma}\label{lemma:baseineq}
If $\bfp$, $\bfq$ and $\bfs$ are positive real vectors, of which $\bfp$ is a increasing vector and  $\bfq$ is a decreasing vector, then
\begin{enumerate}
\item\label{lemma2_1} $\|\bfq\|_2^2\bfp^T\bfs-\bfp^T\bfq\bfq^T\bfs\ge 0$ if $\bfs/\bfq$ is a non-constant increasing vector. The inequality is strict if  $\bfp$ is  {\it strictly} increasing.
\item \label{lemma2_2}$\|\bfq\|_2^2\bfp^T\bfs-\bfp^T\bfq\bfq^T\bfs\le 0$ if $\bfs/ \bfq$ is a non-constant decreasing vector. The inequality is strict if  $\bfp$ is  {\it strictly} increasing.
\item \label{lemma2_3}$\|\bfq\|_2^2\bfp^T\bfs-\bfp^T\bfq\bfq^T\bfs=0$ if $\bfs/ \bfq$ is a constant vector, 
\item\label{lemma2_4} $\|\bfp\|_2^2\bfq^T\bfs-\bfp^T\bfq\bfp^T\bfs\ge 0$ if $\bfs/ \bfp$ is a non-constant decreasing vector. The inequality is strict if  $\bfp$ is  {\it strictly} increasing.
\item \label{lemma2_5}$\|\bfp\|_2^2\bfq^T\bfs-\bfp^T\bfq\bfp^T\bfs\le 0$ if $\bfs/ \bfp$ is a non-constant increasing vector. The inequality is strict if  $\bfp$ is  {\it strictly} increasing.
\item \label{lemma2_6}$\|\bfp\|^2\bfq^T\bfs-\bfp^T\bfq\bfp^T\bfs=0$ if $\bfs/ \bfp$ is a constant vector.
\end{enumerate}
\end{lemma}
\begin{proof}
 We only prove the first case. The proof for the other items follows similarly.
 We use mathematical  induction. 
For the lowest dimension $n=2$, 
 \begin{eqnarray*}
&&\|\bfq\|_2^2\bfp^T\bfs-\bfp^T\bfq\bfq^T\bfs \\
&=&(q_1^2+q_2^2)(p_1s_1+p_2s_2)-(p_1q_1+p_2q_2)(q_1s_1+q_2s_2)\\  
&=&q_1^2p_2s_2+q_2^2p_1s_1 -p_1q_1q_2s_2-p_2q_2q_1s_1\\
&=&(q_1s_2-q_2s_1)(q_1p_2-q_2p_1)\\
 &=&(q_1p_2-q_2p_1)\big((q_1q_2)(\frac{s_2}{q_2}-\frac{s_1}{q_1})\big)\\
&\ge& 0.
\end{eqnarray*}
The last reduction follows from the monotonicity of   $\bfp$, $\bfq$, which implies  $q_1p_2-q_2p_1\ge 0$, and the non-constant increasing assumption of $\bfs/ \bfq$, which guarantees  $\frac{s_2}{q_2}-\frac{s_1}{q_1}>0$. When $\bfp$ is strictly increasing  $q_1p_2-q_2p_1>0$. Under this condition $\|\bfq\|_2^2\bfp^T\bfs-\bfp^T\bfq\bfq^T\bfs>0$ for $n=2$.  Assuming the inequality $\|\bfq\|_2^2\bfp^T\bfs-\bfp^T\bfq\bfq^T\bfs\ge 0$ is true for dimension $n=i$, i.e. $$\displaystyle\sum_{k=1}^iq_k^2\displaystyle\sum_{k=1}^ip_ks_k-\displaystyle\sum_{k=1}^ip_kq_k\displaystyle\sum_{k=1}^iq_ks_k\ge 0,$$
 then for $n=i+1$ 
\begin{eqnarray*}
&&\|\bfq\|_2^2\bfp^T\bfs-\bfp^T\bfq\bfq^T\bfs \\ &=&(\displaystyle\sum_{k=1}^iq_k^2+q_{i+1}^2)(\displaystyle\sum_{k=1}^ip_ks_k+p_{i+1}s_{i+1})-
(\displaystyle\sum_{k=1}^ip_kq_k+p_{i+1}q_{i+1} )(\displaystyle\sum_{k=1}^iq_ks_k+q_{i+1}s_{i+1})\\
&=&(\displaystyle\sum_{k=1}^iq_k^2\displaystyle\sum_{k=1}^ip_ks_k-\displaystyle\sum_{k=1}^ip_kq_k\displaystyle\sum_{k=1}^iq_ks_k )\\
&& +(p_{i+1}s_{i+1}\displaystyle\sum_{k=1}^iq_k^2-q_{i+1}s_{i+1}\displaystyle\sum_{k=1}^ip_kq_k )
+(q_{i+1}^2\displaystyle\sum_{k=1}^ip_ks_k -p_{i+1}q_{i+1} \displaystyle\sum_{k=1}^iq_ks_k)\\
&\ge & 0+ s_{i+1}\displaystyle\sum_{k=1}^iq_k(q_kp_{i+1} -p_kq_{i+1})+q_{i+1} \displaystyle\sum_{k=1}^is_k(q_{i+1}p_k-p_{i+1}q_k)\\
&=&\displaystyle\sum_{k=1}^i (q_kp_{i+1} -p_kq_{i+1})(q_ks_{i+1}-q_{i+1}s_k)\\
&=&\sum_{k=1}^i \big((q_kp_{i+1}-p_kq_{i+1})(q_kq_{i+1})\big)(\frac{s_{i+1}}{q_{i+1}}-\frac{s_k}{q_k}).\\
&\ge & 0.
\end{eqnarray*}
The last reduction is based on the monotonicity of  $\bfp, \bfq$ and $\bfs/ \bfq$. 
 When $\bfp$ is strictly increasing  $q_kp_{i+1} -p_kq_{i+1}>0$ for all $k\le i$ the inequality will be strict because   at least one of the terms $\frac{s_{i+1}}{q_{i+1}}-\frac{s_k}{q_k}, k=1,\cdots,i,$ is positive based on the monotonicity condition. The result thus follows by induction  for all integers $n\ge 2$.
\end{proof}

The following corollary now follows immediately by observing that $\bfs/\bfq$ increases  when $\bfs$ increases and $\bfs/\bfp$ decreases when $\bfs$ decreases.
\begin{cor}\label{cor:ineq_monos}
If $\bfp$, $\bfq$ and $\bfs$ are positive real vectors, of which $\bfp$ is a strictly increasing vector and  $\bfq$ is a decreasing vector, then
\begin{enumerate}
\item $\|\bfp\|_2^2\bfq^T\bfs-\bfp^T\bfq\bfp^T\bfs>0$ if $\bfs$ is a decreasing vector.
\item $\|\bfq\|_2^2\bfp^T\bfs-\bfp^T\bfq\bfq^T\bfs> 0$ if $\bfs$ is an increasing vector.
\end{enumerate}
\end{cor}
% \begin{proof}
% Although these results are a direct corollary of Lemma~\ref{lemma:baseineq} (\ref{lemma2_4}) and (\ref{lemma2_1}), we present a much simpler method of proof which uses Chebychev's  sum inequality twice and the inequality $\sqrt{n}\|\bfx\|_2 \ge \|\bfx\|_1$. We only prove the first case, the second follows similarly. Because both $\|\bfp\|_2^2\bfq-\bfp^T\bfq\bfp$ and $\bfs$ are decreasing we have
% \begin{eqnarray*}
% \|\bfp\|_2^2\bfq^T\bfs-\bfp^T\bfq\bfp^T\bfs&=& (\|\bfp\|_2^2\bfq-\bfp^T\bfq\bfp)^T\bfs \\ 
% &>&\frac{\|\bfs\|_1 }{ n} (\|\bfp\|_2^2\|\bfq\|_1-\bfp^T\bfq\|\bfp\|_1)\\
% &>&\frac{\|\bfs\|_1 }{ n} (\|\bfp\|_2^2\|\bfq\|_1-\frac{ \|\bfp\|_1\|\bfq\|_1 }{ n} \|\bfp\|_1)\\
% &=&\frac{\|\bfs\|_1\|\bfq\|_1 }{ n} (\|\bfp\|_2^2-\frac{ \|\bfp\|_1^2 }{ n} )\\
% &\ge& 0.
% \end{eqnarray*}
% Now $\|\bfp\|_2^2-{ \|\bfp\|_1^2 }/{ n}=0$ if and only if $\bfp$ is  constant, which is not the case and the inequality holds.
% \qed
% \end{proof}

\begin{lemma}\label{lemma:MA1_ana}
If $\bfp$, $\bfq$, $\bfr$ and $\bfs$ are positive real vectors, of which $\bfp$ is strictly increasing, $\bfq$ is decreasing, and $\bfp$, $\bfr$, $\bfs$ and $x^*$  satisfy $\bfp x^*-\bfs=\bfr$; and $[\hat{x},\hat{b}]=\mathrm{argmin}\|\bfp x-b\bfq-\bfr\|_2^2$; then
\begin{enumerate}
 \item the estimated solution $\hat{x}$ and exact solution $x^*$ are related by
\begin{itemize}
\item $\hat{x}>x^*$ if  $\bfs/ \bfq$ is a non-constant decreasing vector,
\item $\hat{x}<x^*$ if  $\bfs/ \bfq$ is a non-constant increasing vector,
\item $\hat{x}=x^*$ if  $\bfs/ \bfq$ is a constant vector;
\end{itemize}
\item  the following inequality is true without any monotonicity assumptions:
\begin{equation}\label{ieq:error_bound_base}
 |\hat{x}-x^*| \le  \frac{\bfp^T\bfq \|\bfq\|_2^2}
{\|\bfp\|_2^2\|\bfq\|_2^2-(\bfp^T\bfq)^2}V(\bar{\bfs}). 
\end{equation}

\item\label{lemma_bsign} the sign of the intercept  $\hat{b}$ is determined as follows:
\begin{itemize}
 \item $\hat{b}>0$ if  $\bfs/ \bfp$ is a non-constant decreasing vector,
\item $\hat{b}<0$ if  $\bfs/ \bfp$ is a non-constant increasing vector,
\item $\hat{b}=0$ if  $\bfs/ \bfp$ is a constant vector;
\end{itemize}

\item given $x=x^*$, the LS solution of $\bfp x-b\bfq\approx\bfr$ for $b$ is $b=\bfq^T\bfs/\|\bfq\|_2^2$;
\item given $b=\bfq^T\bfs/\|\bfq\|_2^2$, the LS solution of $\bfp x-b\bfq\approx\bfr$ for $x$  and the true solution $x^*$ have the same relationship as stated in the first conclusion of this theorem.
\end{enumerate}
\end{lemma}
\begin{proof}
  It is easy to verify that the LS solution of $\bfp x-b\bfq\approx \bfr$  is 
$$\hat{x}=\frac{\|\bfq\|_2^2\bfp^T\bfr-\bfp^T\bfq\bfq^T\bfr}{\|\bfp\|_2^2\|\bfq\|_2^2-(\bfp^T\bfq)^2} ,\quad \hat{b}=\frac{-\|\bfp\|_2^2\bfq^T\bfr+\bfp^T\bfq\bfp^T\bfr}{\|\bfp\|_2^2\|\bfq\|_2^2-(\bfp^T\bfq)^2}.$$ 
The proof then follows as outlined below: 
\begin{enumerate}
\item  Replace $\bfr$ in the  expression for $\hat{x}$ with $\bfp x^*-\bfs$. Then
\begin{eqnarray}
 \hat{x}&=&\frac{\|\bfq\|_2^2\bfp^T(\bfp x^*-\bfs)-\bfp^T\bfq\bfq^T(\bfp x^*-\bfs)}{\|\bfp\|_2^2\|\bfq\|_2^2-(\bfp^T\bfq)^2}\nonumber \\
&=&x^*+\frac{\bfp^T\bfq\bfq^T\bfs-\|\bfq\|_2^2\bfp^T\bfs
}{\|\bfp\|_2^2|\bfq\|_2^2-(\bfp^T\bfq)^2}, \label{eq:hatx}
\end{eqnarray}
and the results immediately follow from  Lemma \ref{lemma:baseineq} (\ref{lemma2_1})-(\ref{lemma2_3}) and the fact $\|\bfp\|_2\|\bfq\|_2>\bfp^T\bfq$ when $\bfp$ is not linear proportional to $\bfq$.

\item
Because 
\begin{eqnarray*}
&&\bfp^T\bfq\bfq^T\bfs-\|\bfq\|_2^2\bfp^T\bfs\\
&=&\bfp^T\bfq (\bfq\circ \bfq)^T \bar{\bfs}-\|\bfq\|_2^2 (\bfp\circ \bfq)^T \bar{\bfs}\\
&\le& \bfp^T\bfq \|\bfq\|_2^2 \max_i( \bar{s}_i)-\|\bfq\|_2^2 \bfp^T\bfq \cdot \min_i(\bar{s}_i)\\
&=&\bfp^T\bfq \|\bfq\|_2^2  (\max_i(\bar{s}_i)-\min_i(\bar{s}_i)),
\end{eqnarray*}
and similarly
$$\bfp^T\bfq\bfq^T\bfs-\|\bfq\|_2^2\bfp^T\bfs   
\ge \bfp^T\bfq \|\bfq\|_2^2 (\min_i(\bar{s}_i)-\max_i(\bar{s}_i)),$$
We have 
$$|\bfp^T\bfq\bfq^T\bfs-\|\bfq\|_2^2\bfp^T\bfs|\le \bfp^T\bfq \|\bfq\|_2^2 (\min_i(\bar{s}_i)-\max_i(\bar{s}_i)).$$ %V(\bar{\bfs}).$$
Using the fact $\|\bfp\|_2^2\|\bfq\|_2^2-(\bfp^T\bfq)^2>0$ and (\ref{eq:hatx}), we conclude the inequality is true.

\item Again we replace $\bfr$  with $\bfp x^*-\bfs$, then the expression for $\hat{b}$ becomes
\begin{eqnarray}
 \hat{b}&=&\frac{-\|\bfp\|_2^2\bfq^T(\bfp x^*-\bfs)+\bfp^T\bfq\bfp^T(\bfp x^*-\bfs)}{\|\bfp\|_2^2|\bfq\|_2^2-(\bfp^T\bfq)^2}\nonumber\\
&=& \frac{\|\bfp\|_2^2\bfq^T\bfs-\bfp^T\bfq\bfp^T\bfs}{\|\bfp\|_2^2\|\bfq\|_2^2-(\bfp^T\bfq)^2}.\label{eq:hatb}
\end{eqnarray}
The results immediately follow from  Lemma \ref{lemma:baseineq} (\ref{lemma2_4})-(\ref{lemma2_6}) and the fact $\|\bfp\|_2\|\bfq\|_2>\bfp^T\bfq$ when $\bfp$ and $\bfq$ do not have the same direction.
\item This result is easily verified.
\item Given $b=\bfq^T\bfs/\|\bfq\|_2^2$, the LS solution of $\bfp x-b\bfq\approx\bfr$ for $x$ is
\begin{eqnarray*}
 \hat{x}&=&\frac{1 }{ \|\bfp\|_2^2}\bfp^T(\bfq b+\bfr)\\
&=&\frac{1 }{ \|\bfp\|_2^2} (\bfp^T\bfq \frac{\bfq^T\bfs }{ \|\bfq\|_2^2}+\bfp^T(\bfp x^*-\bfs) )\\
&=&x^*+ \frac{\bfp^T\bfq\bfq^T\bfs-\|\bfq\|_2^2\bfp^T\bfs }{\|\bfp\|_2^2\|\bfq\|_2^2}.
\end{eqnarray*}
The results now follow from Lemma \ref{lemma:baseineq}.
\end{enumerate}
\end{proof}

We now transform the exact equation to $\bfp/\bfq x^*- \bfs/\bfq=\bfr/\bfq$ and rewrite the results using vectors $\bar{\bfp}=\bfp/\bfq$, $\bar{\bfs}=\bfs/\bfq$ and $\bar{\bfr}=\bfr/\bfq$.  Correspondingly, we find the  LS solution of $\bar{\bfp}x -\bfe b\approx\bar{\bfr}$  for $\bfe=(1$, $1, \cdots, 1)^T$. %In addition,  we define the variation of a function ( or vector) as $V(\bfx(t))=|\max_t\bfx(t)-\min_t \bfx(t)|.$  We have to following conclusions.

\begin{cor}\label{cor:Logan-GA}
If $\bar{\bfp}$, $\bar{\bfr}$ and  $\bar{\bfs}$  are positive, of which  $\bar{\bfp}$ is strictly increasing, $\bar{\bfp}$, $\bar{\bfr}$, $\bar{\bfs}$ and $x^*$  satisfy $\bar{\bfp} x^*-\bar{\bfs}=\bar{\bfr}$; and $[\hat{x},\hat{b}]=\mathrm{argmin}\|\bar{\bfp} x-b\bfe-\bar{\bfr}\|_2^2$, then
\begin{enumerate}
\item \label{item:estx} the estimated solution $\hat{x}$ and the exact solution $x^*$ are related by
\begin{itemize}
\item $\hat{x}>x^*$ if  $\bar{\bfs}$ is a non-constant decreasing vector,
\item $\hat{x}<x^*$ if  $\bar{\bfs}$ is a non-constant increasing vector,
\item $\hat{x}=x^*$ if  $\bar{\bfs}$ is a constant vector;
\end{itemize}
Moreover, the following inequality is true without any monotonicity assumptions.
\begin{equation}\label{DV_error}
 |\hat{x}-x^*|\le \frac{n\|\bar{\bfp}\|_1}{n\|\bar{\bfp}\|_2^2-\|\bar{\bfp}\|_1^2} V(\bar{\bfs}).
\end{equation}
\item \label{item:est_b} The sign of the intercept  $\hat{b}$ is determined as follows:
\begin{itemize}
\item $\hat{b}>0$ if  $\bar{\bfs}/\bar{\bfp}$ is a non-constant decreasing vector,
\item $\hat{b}<0$ if  $\bar{\bfs}/\bar{\bfp}$ is a non-constant increasing vector,
\item $\hat{b}=0$ if  $\bar{\bfs}/\bar{\bfp}$ is a constant vector.
\end{itemize}

In addition, 
\begin{itemize}
 \item $\hat{b}> \displaystyle\sum_{i=1}^n \bar{s}_i/n$  if $\bar{\bfs}$ is a non-constant decreasing vector,
\item $\hat{b}< \displaystyle\sum_{i=1}^n \bar{s}_i/n$ if  $\bar{\bfs}$ is a non-constant increasing vector,
\item $\hat{b}= \displaystyle\sum_{i=1}^n \bar{s}_i/n$ if  $\bar{\bfs}$ is a constant vector;
\end{itemize}
\item Given $x=x^*$, the LS solution of $\bar{\bfp} x-b\bfe\approx\bar{\bfr}$ for $b$ is $b=\displaystyle\sum_{i=1}^n \bar{s}_i/n$;
\item Given $b=\displaystyle\sum_{i=1}^n \bar{s}_i/n$, the LS solution of $\bar{\bfp} x-b\bfe\approx\bar{\bfr}$ for $x$  and the true solution $x^*$ are related as stated in the first conclusion of this theorem.
\end{enumerate}
\end{cor}

\begin{proof}
Most results are a direct Corollary of Lemma \ref{lemma:MA1_ana} by setting $\bfq=\bfe$. We only prove the new results (\ref{item:estx}) and (\ref{item:est_b}).
\begin{enumerate}
\item We just need to prove the bounds for $|\hat{x}-x^*|$. Setting $\bfq=\bfe$ in (\ref{eq:hatx}) we have
\begin{equation}\label{eq:xhat}
 \hat{x}=x^*+\frac{\|\bar{\bfp}\|_1\displaystyle\sum_i \bar{s}_i-n\bar{\bfp}^T\bar{\bfs}
}{n\|\bar{\bfp}\|_2^2-\|\bar{\bfp}\|_1^2}.
\end{equation}
Because 
%\begin{eqnarray*}
$$ \|\bar{\bfp}\|_1\displaystyle\sum_i \bar{s}_i-n\bar{\bfp}^T\bar{\bfs}\le n\cdot \max_i( \bar{s}_i) \|\bar{\bfp}\|_1- n\cdot \min_i(\bar{s}_i) \|\bar{\bfp}\|_1=n\|\bar{\bfp}\|_1 (\max_i(\bar{s}_i)-\min_i(\bar{s}_i)),$$
$$\|\bar{\bfp}\|_1\displaystyle\sum_i\bar{s}_i-n\bar{\bfp}^T\bar{\bfs}\ge n\cdot \min_i(\bar{s}_i) \|\bar{\bfp}\|_1- n\cdot \max_i(\bar{s}_i) \|\bar{\bfp}\|_1=n\|\bar{\bfp}\|_1 (\min_i(\bar{s}_i)-\max_i(\bar{s}_i)),$$
%\end{eqnarray*}
and $n\|\bar{\bfp}\|_2^2-\|\bar{\bfp}\|_1^2>0$  we obtain
\begin{eqnarray*}
|\hat{x}-x^*|&=&\frac{| \|\bar{\bfp}\|_1\displaystyle\sum_i\bar{s}_i-n\bar{\bfp}^T\bar{\bfs} |
}{n\|\bar{\bfp}\|_2^2-\|\bar{\bfp}\|_1^2}\\
&\le& \frac{n\|\bar{\bfp}\|_1(\max_i(\bar{s}_i)-\min_i(\bar{s}_i))
}{n\|\bar{\bfp}\|_2^2-\|\bar{\bfp}\|_1^2}\\
&=&
\frac{n\|\bar{\bfp}\|_1}{n\|\bar{\bfp}\|_2^2-\|\bar{\bfp}\|_1^2} V(\bar{\bfs}).
\end{eqnarray*}
\item  Setting $\bfq=\bfe$ in (\ref{eq:hatb}) we have 
\begin{equation*}
\hat{b}= \frac{\|\bar{\bfp}\|_2^2 \| \bar{\bfs}\|_1 -  \|\bar{\bfp}\|_1  \bar{\bfp}^T \bar{\bfs}} {n\|\bar{\bfp}\|_2^2  -  \|\bar{\bfp}\|_1^2}.
\end{equation*}
The results on the sign follow from Lemma \ref{lemma:MA1_ana} (\ref{lemma_bsign}).
For the remaining three inequalities, we only prove the case for which  $\bar{\bfs}$ is decreasing. Proofs of the other two are similar. Setting $\bfq=\bfe$ in (\ref{eq:hatb}) we have 
\begin{eqnarray*}
\hat{b}&=&
 \frac{\|\bar{\bfp}\|_2^2  \|\bar{\bfs}\|_1 -  \|\bar{\bfp}\|_1  \bar{\bfp}^T \bar{\bfs}} {n\|\bar{\bfp}\|_2^2  -  \|\bar{\bfp}\|_1^2}\\
&>&\frac{ \|\bar{\bfp}\|_2^2  \| \bar{\bfs}\|_1 -  \|\bar{\bfp}\|_1 \frac{1 }{ n} \|\bar{\bfp}\|_1 \|\bar{\bfs}\|_1  } {n\|\bar{\bfp}\|_2^2  -  \|\bar{\bfp}\|_1^2}\\
&=&\frac{\|\bar{\bfs}\|_1  }{ n}=\frac{\displaystyle\sum_{i=1}^n \bar{s}_i }{ n}.
\end{eqnarray*}
\end{enumerate}
\end{proof}

\section{Appendix B: component-wise perturbation analysis for LS solution of (\ref{eq:LSnewmodel})}
In Remark 2,  we claimed that ``the estimate of DV is much more robust to noise in the formulation than are the estimates of $A$ and $B$ because $\int_0^{t}\Cp(\tau){\rm d}\tau$ is much larger than both $\CT(t)$ and $\CS(t)$ for $t>t^*$''. Here we present a theoretical explanation, which is helpful for algorithm design in quantification. Instead of considering a general linear equation, which is out of the range of this paper, we assume a system of
 equations $A\bfx=\bfy$ with only  two independent variables $\bfx=[x_1,x_2]^T$. The two columns of the system matrix $A$ are denoted by $\bfa_1$ and $\bfa_2$, i.e. $A=[\bfa_1, \bfa_2]$. 
\begin{theorem} \label{thm:var}
Suppose the linear system $A\bfx \approx \bfy+{\mathbf \epsilon}$, for $A=[\bfa_1, \bfa_2]$,  has the exact solution $\bfx=[x_1^*, x_2^*]$, the uncorrelated noise vector ${\mathbf \epsilon}$ obeys a multi-variable Gaussian distribution with zero means and common variance $\sigma^2$ and that $\|\bfa_1\|>>\|\bfa_2\|$. Then least squares  solution $\hat{\bfx}=[\hat{x_1},\hat{x_2}]^T$ has the following statistical properties
\begin{enumerate}
 \item $E(\hat{x}_1)=x_1^*$ and $E(\hat{x}_2)=x_2^*$, and
\item $\mathrm{Var}(\hat{x}_1) << \mathrm{Var}(\hat{x}_2)$.
\end{enumerate}
\end{theorem}
\begin{proof}
 We assume matrix $A$ has the following singular value decomposition 
\begin{equation}\label{eq:svd}
A=[\bfa_1, \bfa_2]=U S V^T=U
\left (\begin{array}{cc}
s_1 & 0\\
0& s_2\\
\vdots& \dots\\
0&0
\end{array}
\right )
\left (\begin{array}{cc}
\cos\theta & \sin\theta\\
-\sin\theta&\cos\theta
\end{array}
\right ),
\end{equation}
in which $s_1 \ge s_2$. Then  
$$\hat{\bfx}=VS^{\dag}U^T(\bfy+{\mathbf \epsilon})=\bfx^*+VS^{\dag}U^T {\mathbf \epsilon},$$
where 
$$S^{\dag}=  \left (\begin{array}{cccc}
1/s_1 & 0&\cdots&0\\
0& 1/s_2 &\cdots&0
\end{array}
\right ).$$ 
Because $U$ is an unitary matrix  and $\|\bfa_1\|>>\|\bfa_2\|$ we immediately derive  the the following inequality from equation (\ref{eq:svd}):
$$s_1^2\cos^2\theta+s_2^2\sin^2\theta>> s_1^2\sin^2\theta+s_2^2\cos^2\theta.$$
This inequality is equivalent to $(s_1^2-s_2^2)\cos^2\theta+s_2^2>> (s_1^2-s_2^2)\sin^2\theta+s_2^2,$ which implies $\cos^2\theta>>\sin^2\theta$, i.e. $\cos^2\theta\approx 1$ and $\sin^2\theta\approx 0$, and $s_1^2>>s_2^2$. If we denote the two rows of matrix $VS^{\dag}$ by $\bfq_1$ and $\bfq_2$ than 
\begin{eqnarray*}
 \|\bfq_1\|^2&=&\sin^2\theta/s_1^2+\cos^2\theta/s_2^2=1/s_1^2+\cos^2\theta(1/s_2^2-1/s_1^2),\\  
\|\bfq_2\|^2&=&\cos^2\theta/s_1^2+\sin^2\theta/s_2^2=1/s_1^2+\sin^2\theta(1/s_2^2-1/s_1^2).\\  
\end{eqnarray*}
 Because $\cos^2\theta>>\sin^2\theta$ and $1/s_2^2>>1/s_2^2$ we conclude $\|\bfq_2\|^2>>\|\bfq_1\|^2.$  If we let $\bfp_1$ and $\bfp_2$ be the two rows of matrix $VS^{\dag}U^T$ then $\|\bfp_1\|=\|\bfq_1\|$ and $\|\bfp_2\|=\|\bfq_2\|$ because $U$ is unitary. Thus $\|\bfp_2\|^2>>\|\bfp_1\|^2.$ Let 
$$\bfd=\hat{\bfx}-\bfx^*=VS^{\dag}U^T {\mathbf \epsilon}.$$
It is clear $E(d_1)=0$ and $E(d_2)=0$ because the means of ${\mathbf \epsilon}$ are zero, and  $\mathrm{Var}(d_1)=\sum_i p_{1i}^2\sigma^2=\|\bfp_1\|^2\sigma^2$ and $\mathrm{Var}(d_2)=\sum_i p_{2i}^2\sigma^2=\|\bfp_2\|^2\sigma^2$ resp.. Therefore $\mathrm{Var}(\hat{d_1}) << \mathrm{Var}(\hat{d_2})$. 
Because $\bfd=\hat{\bfx}-\bfx^*$ we conclude $E(\hat{\bfx})=\bfx^*$ and  $Var(\hat{x}_1) << \mathrm{Var}(\hat{x}_2)$
\end{proof}
This result is illustrated by the following simple example:
$$
\left (\begin{array}{cc}
4 & 1\\
8& 1\\
10& 1
\end{array}
\right )
\bfx =
\left (\begin{array}{c}
5\\
9\\
11
\end{array}
\right )
+
\left (\begin{array}{l}
\epsilon_1\\
\epsilon_2\\
\epsilon_3
\end{array}
\right )
$$
 The first column is much larger than the second column. If we add $1\%$ noise to the right hand side, i.e. $\epsilon_1\sim N(0, 0.05), \epsilon_2\sim N(0, 0.09)$ and $\epsilon_3\sim N(0, 0.115)$, and perform simulation with $1000$ realizations the distribution of the resulted $x_1$ and $x_2$ are illustrated in Figure~\ref{fig:diffsens}. These results are consistent with the conclusions in Theorem \ref{thm:var}.
\begin{figure}[bt]
\centerline{\includegraphics[height=3in,width=4in]{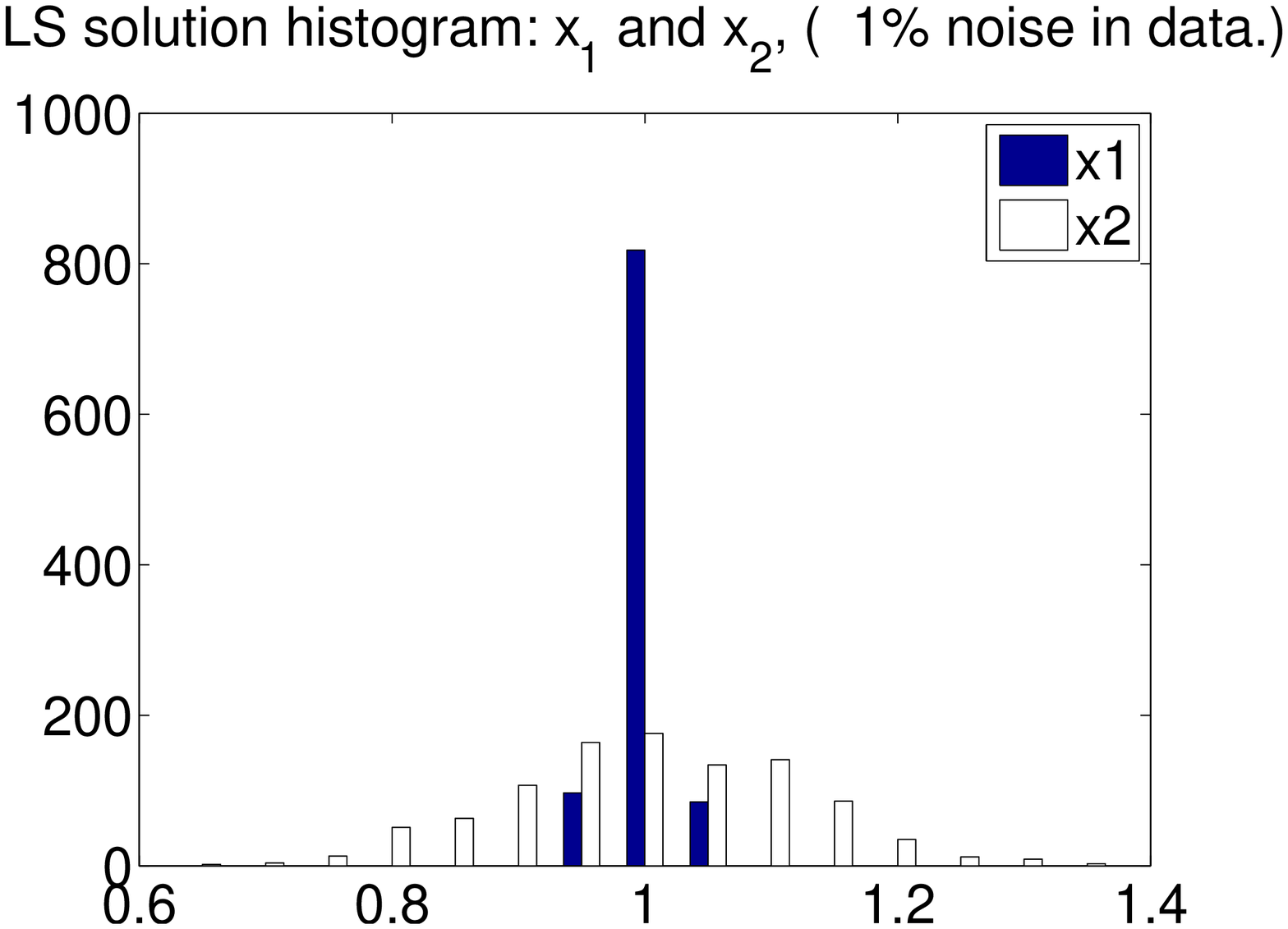}}
\label{fig:diffsens}
\end{figure}

\section{Appendix C: derivation for equation (\ref{eq:root})}
Integrating (\ref{eq:2T4kC1}) and (\ref{eq:2T4kCS}) from $0$ to $t$ we obtain
\begin{eqnarray}
 C_{F+NS}(t)&=&K_1\int_0^t \Cp(\tau){\rm d}\tau-(k_2+k_3)\int_0^tC_{F+NS}(\tau){\rm d}\tau+k_4\int_0^t\CS(\tau){\rm d}\tau, \label{eq:int2T4kC1} \\
 \CS(t)&=&k_3 \int_0^tC_{F+NS}(\tau){\rm d}\tau-k_4\int_0^t\CS(\tau){\rm d}\tau,\label{eq:int2T4kCS}\\
       &=& k_3 \int_0^tC_{F+NS}(\tau){\rm d}\tau-k_4\int_0^t(\CT(\tau)-C_{F+NS}(\tau)) {\rm d}\tau,\nonumber\\
       &=& -k_4 \int_0^t\CT(\tau) {\rm d}\tau +(k_3+k_4) \int_0^tC_{F+NS}(\tau) {\rm d}\tau.
 \label{eq:int2T4kCS1}
\end{eqnarray}
Taking the sum of equations  (\ref{eq:int2T4kC1}) and (\ref{eq:int2T4kCS}) yields: 
\begin{equation}\label{eq:CTsum}
 C_T(t)=K_1\int_0^t \Cp(\tau){\rm d}\tau-k_2\int_0^tC_{F+NS}(\tau){\rm d}\tau,
\end{equation}
and canceling $\int_0^tC_{F+NS}(\tau){\rm d}\tau$ from (\ref{eq:int2T4kCS1}) using (\ref{eq:CTsum}) gives:
$$  
\CS(t)= -k_4 \int_0^t\CT(\tau) {\rm d}\tau +\frac{k_3+k_4}{k_2}\left ( K_1\int_0^t \Cp(\tau){\rm d}\tau-C_T(t)\right ).
$$
This  can be transformed to (\ref{eq:root}) immediately by using $DV=\frac{K_1}{k_2}(1+\frac{k_3}{k_4})$.

\end{document}